\documentclass[aps,twocolumn,superscriptaddress,nofootinbib,a4paper,longbibliography,pra]{revtex4-2}

\usepackage{graphicx, subfigure}
\usepackage{float}
\usepackage{dcolumn}
\usepackage{bm}
\usepackage[colorlinks=true,linkcolor=blue,citecolor=magenta,urlcolor=blue]{hyperref}
\usepackage{physics}
\usepackage{amsthm}
\usepackage{bbold}
\usepackage{thm-restate}
\usepackage{amsmath, amssymb, amsfonts, amsthm, bm, bbold, physics}
\usepackage{mathtools}
\usepackage{xcolor, graphicx, subfigure, float, hyperref}
\usepackage{tcolorbox, changepage}
\usepackage{natbib}
\usepackage{dsfont}
\usepackage{thm-restate}
\usepackage[normalem]{ulem}
\usepackage{lineno}

\newtheorem{lem}{Lemma}
\newtheorem{cor}{Corollary}

\newtheorem{fact}{Fact}

\newtheorem{defi}{Definition}
\newtheorem{prop}{Proposition}

\newenvironment{italiccenter}
    {\begin{center}\em}
    {\end{center}}

\begin{document}


\title{Test of the physical significance of Bell non-locality}


\author{Carlos Vieira}
\email{carloshv@unicamp.br}
\affiliation{Department of Computer Science, The University of Hong Kong, Pokfulam Road, Hong Kong}
\affiliation{Instituto de Matemática, Estatística e Computação Científica, Universidade Estadual de Campinas, CEP 13083-859, Campinas, Brazil}

\author{Ravishankar Ramanathan}
\email{ravi@cs.hku.hk}
\affiliation{Department of Computer Science, The University of Hong Kong, Pokfulam Road, Hong Kong}

\author{Ad\'an Cabello}
\email{adan@us.es}
\affiliation{Departamento de F\'{\i}sica Aplicada II, Universidad de Sevilla, E-41012 Sevilla,
Spain}
\affiliation{Instituto Carlos~I de F\'{\i}sica Te\'orica y Computacional, Universidad de
Sevilla, E-41012 Sevilla, Spain}


\date{\today}



\begin{abstract}
Loophole-free violations of Bell inequalities imply that at least one of the assumptions behind local hidden-variable theories must fail. Here, we show that, if only one fails, then it has to fail completely, therefore excluding models that partially constrain freedom of choice or allow for partial retrocausal influences, or allow partial instantaneous actions at a distance. Specifically, we show that (i) any hidden-variable theory with outcome independence (OI) and arbitrary joint relaxation of measurement independence (MI) and parameter independence (PI) can be experimentally excluded in a Bell-like experiment with many settings on high-dimensional entangled states, and (ii) any hidden-variable theory with MI, PI and arbitrary relaxation of OI can be excluded in a Bell-like experiment with many settings on qubit-qubit entangled states.
\end{abstract}


\maketitle
\twocolumngrid




In the early days of quantum theory, the question of whether there is
deeper theory underlying quantum theory was considered ``a philosophical question for which physical arguments alone are not decisive'' \cite{Born:1926ZPb}. Bell's theorem \cite{Bell:1964PHY,Clauser:1969PRL} made it possible to exclude experimentally some of these deeper theories, called hidden-variable (HV) theories \cite{Einstein:1935PR}. Today, Bell tests \cite{Freedman:1972PRL,Aspect:1982PRL,Hensen:2015NAT, Giustina:2015PRL, Shalm:2015PRL} have convinced us that {\em some} HV theories cannot explain what we see.


In a Bell test, a source of pairs of particles sends each particle to a different laboratory. In the first laboratory, an observer (Alice) chooses to measure $x \in X$ and obtains $a \in A$. In the second laboratory, a different observer (Bob) chooses to measure $y \in Y$ and obtains $b \in B$. After many repetitions, Alice and Bob compute the joint probability of $(a, b)$ given $(x, y)$, denoted $p(a,b \vert x,y)$. The set $\{p(a,b \vert x,y)\}_{x \in X, y \in Y, a \in A, b \in B}$ is called a {\em correlation} for the {\em Bell scenario} $(|X|, |A|; |Y|, |B|)$, in which Alice can choose between $|X|$ measurement settings with $|A|$ possible outcomes and Bob between $|Y|$ settings with $|B|$ outcomes.

Bell's theorem asserts that no HV model satisfying some assumptions can reproduce certain quantum correlations. 
These models are collectively called ``local'' HV models and are defined as those satisfying the following assumptions \cite{jarrett_physical_1984, Shimony:1993}:

\medskip

(0) {\em Hidden variables}. There are hidden variables that associate to each pair of particles a state $\lambda \in \Lambda$ and underlying probability densities $p(a,b \vert \lambda,x,y)$ and $p (\lambda \vert x,y)$ so 
\begin{equation}\label{eq:HiddenVariableModel}
    p(a,b \vert x,y) = \int d\lambda\, p (\lambda \vert x,y) p(a,b \vert \lambda,x,y).
\end{equation}

(1) {\em Measurement independence} (MI): For every pair of particles, the measurements $(x,y)$ are not correlated with $\lambda$. That is, $ p(x,y \vert \lambda) = p(x,y)$,
which, through Bayes's theorem, is equivalent to
\begin{equation}
    p(\lambda \vert x,y) = p(\lambda).
\end{equation}
Therefore, the knowledge of $\lambda$ gives no information about $(x,y)$, and vice versa. 

\medskip

(2) {\em Outcome independence} (OI) \cite{Shimony:1993} also referred to as completeness \cite{jarrett_physical_1984,Hall:2015XXX}: 
$p(a \vert \lambda,x,y,b)$ is independent of $b$, and hence may be written 
\begin{equation}
    p(a \vert \lambda,x,y,b) = p(a \vert \lambda,x,y).
\end{equation}
Similarly, $p(b \vert \lambda,x,y,a) = p(b \vert \lambda,x,y)$.

(3) {\em Parameter independence} (PI) \cite{Shimony:1993}, initially called locality \cite{jarrett_physical_1984}:  $p(a \vert \lambda,x,y)$ is independent of $y$, and hence may be written as 
\begin{equation}
    p(a \vert \lambda,x,y) = p(a \vert \lambda,x).
\end{equation}
Similarly, $p(b \vert \lambda, x,y)=p(b \vert \lambda, y)$.

Assumptions (2) and (3) are independent \cite{jarrett_physical_1984} and, together, imply that  
\begin{equation}
p(a,b \vert \lambda,x,y) = p(a \vert \lambda, x) p(b|\lambda,y),
\end{equation}
which is Bell's original assumption \cite{Bell:1964PHY} which is now called local factorizability or local causality.

Assumption (0) is the expression of the belief in a deeper theory underlying quantum theory.  MI is motivated by the assumption that each of the observers has freedom of choice \cite{Bell:1985Dia} or, more generally, by the assumption that which specific measurements are actually performed is not governed by the HVs that govern the particles.
OI is based on the assumption that, as it happens in deterministic models [i.e., when $p(a,b|\lambda,x,y) \in \{0,1\}]$, if we would know $\lambda$, we would observe that $p(a,b \vert \lambda, x,y) = p(a \vert  \lambda,x,y) p(b \vert  \lambda,x,y)$ \cite{hall_relaxed_2011}.
PI is grounded on the assumption that superluminal signalling between one party's choice and the other party's spacelike separated outcome is impossible \cite{jarrett_physical_1984}.

Existing experiments are inconclusive about which assumptions fail. As a consequence, possible explanations include HV theories with different degrees of ``measurement dependence''  \cite{colbeck_free_2012, brandao_realistic_2016, hall_complementary_2010, barretthow2011,hall_relaxed_2011, Hall:2015XXX, ramanathan2016randomness, hall_measurement-dependence_2020, ramanathan2021nosignalingproof} (that may occur due to limitations to freedom of choice \cite{Brans:1988IJTP,tHooft:2016} or to retrocausal influences \cite{Costa:1953MQ,donadi_toy_2022}), different amounts of instantaneous ``actions at a distance'' \cite{Bohm:1952PR,Brassard:1999PRL,Bacon:PRL2003,P2003PRA, pawlowski_non-local_2010,Ringbauer:2016SA,brask_bell_2017}, and combinations thereof \cite{blasiak_violations_2021}. At least one of the four assumptions is false. But which one or which ones? \cite{Shimony:1993}[pp.~124,
149,
96], 
\cite{Shimony:1993a}[p.~66].
The prevalent view is that advancing in the resolution of this problem is not possible ``on purely physical grounds but it requires an act of metaphysical judgement'' \cite{Polkinghorne2014}. Here, we challenge this view and present two results. Result~1 shows that there are quantum correlations that cannot be simulated with any HV theory assuming OI but {\em partial} (as opposed to complete) measurement dependence (MD) or {\em partial} parameter dependence (PD). Result~2 shows that there are quantum correlations that cannot be simulated with any HV theory assuming MI, PI but {\em partial} outcome dependence (OD).

In Sec.~\ref{sec:quantification}, we introduce the standard ways to quantify MI, PI and OI. Result 1 is presented in Sec.~\ref{sec:mainp}, where we also describe an experiment to exclude HV theories with partial MD and PD. Result 2 is presented in Sec.~\ref{sec:ResultOI}, which includes the description of an experiment to exclude HV theories with partial OD. The consequences and applications of the results are discussed in Sec.~\ref{sec:conclusions}. 


\section*{Results}


\subsection*{Relaxing the assumptions}
\label{sec:quantification}


{\em Quantifying measurement dependence.} To quantify any lack of MI, and therefore to quantify MD, we have to take into account the distribution of $x$ and $y$ and, therefore, we have to consider the full distribution $p(a,b,x,y)$ rather than only $p(a,b \vert x,y)$. The full distribution $p(a,b,x,y)$ can be reproduced with an $l$-measurement dependent ($l$-MD) HV model \cite{putz_arbitrarily_2014} if it can be reproduced with an HV model such that, for all $x \in X$, $y \in Y$, and for all $\lambda$,
\begin{equation}\label{eq:MDdefinition}
p(x, y \vert \lambda) \ge l \ge 0.
\end{equation}
If there are only two inputs per party, a value $l=1/4$ implies that $p(x, y \vert \lambda)$ must be uniform. Therefore, in this case, $p(a,b,x,y)$ can be reproduced with an HV model with MI in which there is no correlation between the hidden variables of the particles and the measurement settings. However, this is not the case for $0 \le l < 1/4$. We say that $p(a,b,x,y)$ can be reproduced with an HV model with {\em partial} MD if there is an ($l$-MD) model for some $l>0$. We say that $p(a,b,x,y)$ can only be reproduced with an HV model with {\em complete} MD if $p(a,b,x,y)$ cannot be reproduced with any ($l$-MD) model for any $l>0$; the only possible HV models have $p(x,y|\lambda) = 0$ for some pair of settings $(x,y)$ and some $\lambda$. If there are only two inputs per party, any $p(a,b,x,y)$ corresponding to any non-signaling correlation can be reproduced with an HV model with complete MD.
The complete relaxation of MI using alternative ways of quantifying MD \cite{barretthow2011,hall_relaxed_2011} matches the above definition of complete MD (see also Supplementary Note 1).
 
\medskip


{\em Quantifying parameter dependence.} A correlation $p(a,b|x,y)$ can be reproduced with an $(\varepsilon_A, \varepsilon_B)$-parameter dependent [$(\varepsilon_A, \varepsilon_B)$-PD] HV model \cite{hall_relaxed_2011} if it can be reproduced with an HV model such that, 
for all $x$, $y$, $y'$ ($y$, $x$, $x'$), and for all $\lambda$, 
\begin{subequations}\label{eq:ParDepen}
\begin{align}
& \frac{1}{2} \sum_{a} |p\left(a \vert \lambda, x, y\right) - p\left(a \vert \lambda, x, y^{\prime}\right)| \le \varepsilon_A
\label{eq:ParDepenA} \\
& \frac{1}{2}\sum_{b} |p\left(b \vert \lambda, x, y\right) - p\left(b \vert \lambda, x', y\right)| \le \varepsilon_B.
\label{eq:ParDepenB}
\end{align}
\end{subequations}
Therefore, $p(a,b|x,y)$ can be reproduced with an HV model with PI if, and only if, it can be reproduced with an $(\varepsilon_A, \varepsilon_B)$-PD HV model with $\varepsilon_A=\varepsilon_B=0$. If not, we say that $p(a,b|x,y)$ can be reproduced with an HV model with {\em partial} PD if 
it can be reproduced with an $(\varepsilon_A, \varepsilon_B)$-PD HV model for some $0< \varepsilon_A, \varepsilon_B < 1$. Finally, $p(a,b|x,y)$ can only be reproduced with HV models with {\em complete} PD if $p(a,b|x,y)$ cannot be reproduced with any $(\varepsilon_A, \varepsilon_B)$-PD HV model for any $\varepsilon_A, \varepsilon_B < 1$. Any non-signaling correlation can be reproduced with HV models with complete PD.


\medskip

{\em Quantifying outcome dependence.} A correlation $p(a,b|x,y)$ can be reproduced with a $\delta$-outcome dependent ($\delta$-OD) HV model \cite{hall_relaxed_2011} if it can be reproduced with an HV model such that, for all $x$, $y$, $a$, $a'$, and for any $\lambda$, 
\begin{equation}\label{eq:epsilon-OD}
\frac{1}{2} \sum_b \big\vert p(b|\lambda,x,y,a) - p(b|\lambda,x,y,a') \big\vert \leq \delta.
\end{equation}
Therefore, $p(a,b|x,y)$ can be reproduced with an HV model with OI if, and only if, it can be reproduced with a $\delta$-OD HV model with $\delta=0$. 
We say that $p(a,b|x,y)$ can be reproduced with an HV model with {\em partial} OD if it can be reproduced with a $\delta$-OD HV model with
$0< \delta < 1$.  We say that $p(a,b|x,y)$ can only be reproduced with an HV model with {\em complete} OD if it cannot be reproduced with any $\delta$-OD HV model for any $\delta < 1$. Any non-signaling correlation can be reproduced with HV models with complete OD.


\subsection*{Result 1: Quantum correlations that cannot be simulated if there is arbitrarily small MI or PI}
\label{sec:mainp}


Consider the bipartite Bell experiment in which Alice and Bob have two measurement options $x,y \in \{0,1\}$, each of them with $2^N$ possible results which can be expressed as a string of $N$ bits, $a,b \in \{(0,0,\ldots,0),(0,0,\ldots,1) ,\ldots,(1,1,\ldots,1)\}$. Suppose that Alice and Bob share the following $2^N \times 2^N$-dimensional entangled state:
\begin{equation}\label{eq:StateHardyMV}
\ket{\psi} = \ket{\phi}_{A_1,B_1} \otimes \cdots \otimes \ket{\phi}_{A_N,B_N},
\end{equation}
where
\begin{equation}
\label{eq:8}
\ket{\phi} = a(\ket{01} + \ket{10}) + \sqrt{1 - 2a^2}\ket{11},
\end{equation}
with $a= \frac{\sqrt{5}-1}{2}$, is a two-qubit state with the first qubit in Alice's side and the second qubit in Bob's side. Suppose that Alice's and Bob's measurements are of the form
\begin{subequations}\label{eq:MeasHardyMV}
\begin{align}
& A_{a_1,\ldots,a_N \vert x} = A_{a_1 \vert x}\otimes \cdots \otimes A_{a_N \vert x}, \\
& B_{b_1,\ldots,b_N \vert y} = B_{b_1 \vert y}\otimes \cdots \otimes B_{b_N \vert y},
\end{align}
\end{subequations}
where, here, the tensor product refers to the qubits in each observer's system and the specific form of the factors is given by
\begin{subequations}
\begin{align}
& A_{1 \vert x} = \mathds{1} - A_{0 \vert x},\\
& B_{1 \vert y} = \mathds{1} - B_{0 \vert y}, 
\end{align}
\end{subequations}
where
\begin{subequations}\label{eq:MeasHardyMV2}
\begin{align}
& A_{0 \vert 0} = B_{0 \vert 0} = \ketbra{0}{0},\\
& A_{0 \vert 1} = B_{0 \vert 1} = \ketbra{\varphi}{\varphi}, 		
\end{align}
\end{subequations}
with $\ket{\varphi} = \frac{1}{\sqrt{1 - a^2}} (\sqrt{1 - 2a^2}\ket{0} - a\ket{1})$.
That is, each of the $2^N$-outcome measurements can be seen as $N$ (nonindependent) two-outcome measurements performed simultaneously on a $2^N$-dimensional quantum system. 
These state and measurements produce a correlation with the following properties:
\begin{widetext}
\begin{subequations}\label{eq:HardyZero}
\begin{align}
& p(0, 1,a_{2}, b_{2}, \ldots,a_{N}, b_{N} \vert 0, 1)=\ldots=p(a_{1}, b_{1}, \ldots,a_{N-1}, b_{N-1},0,1 \vert 0, 1)=0,
\label{eq:HardyZero01} \\
& p(1, 0, a_{2}, b_{2}, \ldots, a_{N}, b_{N} \vert 1, 0)=\ldots=p(a_{1}, b_{1}, \ldots, a_{N-1}, b_{N-1},1,0 \vert 1,0)=0,
\label{eq:HardyZero10} \\				
& p(0, 0,a_{2}, b_{2}, \ldots,a_{N}, b_{N} \vert 1, 1)=\ldots=p(a_{1}, b_{1}, \ldots,a_{N-1}, b_{N-1},0,0 \vert 1, 1)=0,
\label{eq:HardyZero11}
\end{align}
\end{subequations}
\end{widetext}
for all $a_1,\ldots,a_N,b_1,\ldots,b_N \in \{0,1\}$. Eq.~\eqref{eq:HardyZero01} indicates that, if the measurements are $x=0$ for Alice and $y=1$ for Bob, then, in the $N$-bit strings that Alice and Bob obtain as outputs cannot be one position where Alice has $0$ and Bob has $1$. Similarly, for Eqs.~\eqref{eq:HardyZero10} and \eqref{eq:HardyZero11}. These state and measurements are the ones needed for the parallelised version \cite{mancinska_unbounded_2015} of the optimal version of the proof of Bell non-locality proposed by Hardy \cite{HardyPRL1993}.

Let us define
\begin{align}
\label{eq:P_hardyQuantum}
p_{H}^{N} \coloneqq \sum_{\substack{a_1, \ldots, a_N, b_1, \ldots, b_N \\ \left(a_1, b_1\right)=(0,0) \lor \ldots \lor \left(a_N, b_N\right)=(0,0)}} p\left(a_{1}, b_{1}, \ldots,a_{N}, b_{N} \vert 0, 0\right),
\end{align}
where $\lor$ is the logical OR. 

Result 1 can be stated as follows: In any $l$-MD and $(\varepsilon_A, \varepsilon_B)$-PD HV model satisfying OI and Eqs.~\eqref{eq:HardyZero01}, \eqref{eq:HardyZero10} and
\eqref{eq:HardyZero11}, for all $l >0$ and all $N$, 
\begin{align}\label{eq:UpBoundP_H}
p_{H}^{N} \leq \varepsilon_{A}+\varepsilon_{B}-\varepsilon_{A} \varepsilon_{B}.
\end{align}
The proof is in the Supplementary Note 2.
Therefore, if $\varepsilon_A < 1$ and $\varepsilon_B < 1$, then $p_{H}^{N} < 1$. In contrast, in quantum theory \cite{mancinska_unbounded_2015}, as $N$ tends to infinity, 
\begin{equation}
p_{H}^{N} \stackrel{N \rightarrow \infty}{\longrightarrow} 1.
\end{equation}
Consequently, for any $l$-MD and $(\varepsilon_A, \varepsilon_B)$-PD HV model with $l >0$, $\varepsilon_A < 1$, $\varepsilon_B < 1$, satisfying OI, there is $N$ such that quantum theory predicts a value for $p_{H}^{N}$ that cannot be simulated.

For example, Table~\ref{tab:my_label} gives the values of $\varepsilon =\varepsilon_A = \varepsilon_B$ that cannot be simulated if nature achieves the quantum value for $p_{H}^{N}$. Notice that the number of excluded HV models grows with $N$. As $N$ tends to infinity, the only surviving HV models are those with $\varepsilon =1$.


\begin{table}
\centering
    \begin{tabular}{ccc} \hline \hline
         $N$&$\varepsilon$&   $p_{H}^{N}$\\ \hline 
         1& $<$ 0.0461&   0.0902\\ 
         2& $<$ 0.0901&   0.1722\\ 
         3& $<$ 0.1321&   0.2469\\ 
         4& $<$ 0.1722&   0.3148\\ 
         5& $<$ 0.2104&   0.3766\\ 
         6& $<$ 0.2468&   0.4328\\ 
         7& $<$ 0.2816&   0.4839\\ 
         8& $<$ 0.3147&   0.5304\\ 
         9& $<$ 0.3463&   0.5727\\ 
         10& $<$ 0.3765&  0.6113\\ 
         \hline \hline
    \end{tabular}
    \caption{$N$ is such that $2^N$ is the number of outputs in the Bell test. $\varepsilon = \varepsilon_A = \varepsilon_B$ quantifies the relaxation of PI and $p_{H}^{N}$ is the upper bound in the probability given by Eq.~\eqref{eq:UpBoundP_H} for $l$-MD and $(\varepsilon, \varepsilon)$-PD HV models satisfying OI. HV models with $\varepsilon$ above the threshold indicated in the Table cannot be excluded by the corresponding experiment.}
    \label{tab:my_label}
\end{table}


The correlations defined by Eqs.~\eqref{eq:8}--\eqref{eq:MeasHardyMV2} are special: they define an extremal non-exposed point of the quantum set of correlations for the two-observer two-setting two-outcome Bell scenario \cite{Goh2018PRA}. Any other correlation with this property can also be used in the experiment. The full characterisation of these points is in \cite{ZhaoQ2023}.

A natural question is what conditions a correlation must satisfy to allow for arbitrarily small MI and PI, and whether there are quantum correlations in Bell scenarios with finite number of inputs that allow for such relaxation. In the Supplementary Note 3, we show a necessary condition - the quantum correlation must necessarily lie on or be arbitrarily close to the nonsignaling boundary. We also illustrate by an explicit example that this condition is not sufficient. We leave as an open question whether there is a finite input-output quantum correlation that proves Result~1.


{\em Experimental test to exclude HV theories with partial MD and PD.} So far, we have identified a quantum correlation that cannot be simulated by any $l$-MD and $(\varepsilon_A, \varepsilon_B)$-PD HV model with $l > 0$, $\varepsilon_A < 1$, $\varepsilon_B < 1$, satisfying OI. This correlation is a point in the set of quantum correlations. The problem is that, due to experimental errors, an actual experiment will fail to exactly produce this point. Here, we reformulate Result~1 in a way that the existence of correlations that cannot be simulated by $l$-MD and $(\varepsilon_A, \varepsilon_B)$-PD HV models with $l > 0$, $\varepsilon_A < 1$, $\varepsilon_B < 1$, and satisfying OI, can be experimentally tested.

It can be proven (see Supplementary Note 4) that, for any $l$-MD and $(\varepsilon_A, \varepsilon_B)$-PD HV model with $l > 0$, $\varepsilon_A < 1$, $\varepsilon_B < 1$, satisfying OI, the following Bell-like inequality holds:
\begin{equation} \label{bli}
     I_{\kappa}^{N}(p_{A B \vert XY}) \leq \tilde{\varepsilon}_{A}+ \tilde{\varepsilon}_{B} - \tilde{\varepsilon}_{A}\tilde{\varepsilon}_{B},
\end{equation}
where
\begin{widetext}
\begin{align}\label{eq:Ikappa}
   I_{\kappa}^{N}(p_{A B \vert XY}) &\coloneqq \sum_{\substack{a_1, \ldots, a_N, b_1, \ldots, b_N \\ \left(a_1, b_1\right)=(0,0) \lor \ldots \lor \left(a_N, b_N\right)=(0,0)}} p_{A B \vert XY}\left(\left(a_{1}, b_{1}\right), \ldots,\left(a_{N}, b_{N}\right) \vert 0,0\right) \notag \\
   &- \kappa \sum_{\substack{a_1, \ldots, a_N, b_1, \ldots, b_N \\ \left(a_1, b_1\right)=(0,1) \lor \ldots \lor \left(a_N, b_N\right)=(0,1)}} p_{A B \vert XY}\left(\left(a_{1}, b_{1}\right), \ldots,\left(a_{N}, b_{N}\right) \vert 0,1\right) \notag \\
    & - \kappa \sum_{\substack{a_1, \ldots, a_N, b_1, \ldots, b_N \\ \left(a_1, b_1\right)=(1,0) \lor \ldots \lor \left(a_N, b_N\right)=(1,0)}} p_{A B \vert XY}\left(\left(a_{1}, b_{1}\right), \ldots,\left(a_{N}, b_{N}\right) \vert 1,0\right)\notag\\
    &- \kappa \sum_{\substack{a_1, \ldots, a_N, b_1, \ldots, b_N \\ \left(a_1, b_1\right)=(0,0) \lor \ldots \lor \left(a_N, b_N\right)=(0,0)}} p_{A B \vert XY}\left(\left(a_{1}, b_{1}\right), \ldots,\left(a_{N}, b_{N}\right) \vert 1,1\right),
\end{align}
\end{widetext}
with
\begin{equation}\label{eq:LB_kappaN}
    \kappa > \frac{N^2}{l(1 - \varepsilon)^2},
\end{equation}
where $\varepsilon = \max\{\varepsilon_A, \varepsilon_B\}$, and
\begin{subequations}
\label{ets2}
\begin{align}
&\tilde{\varepsilon}_A = \varepsilon_{A}+ N\sqrt{\frac{2}{l\kappa}}, \\
& \tilde{\varepsilon}_B = \varepsilon_{B}+ N\sqrt{\frac{2}{l\kappa}}.
\end{align}
\end{subequations}
This means that, for any $l$-MD and $(\varepsilon_A, \varepsilon_B)$-PD HV model with $l > 0$, $\varepsilon_A < 1$, $\varepsilon_B < 1$, satisfying OI, for sufficiently large $\kappa$, the quantity
$I_{\kappa}^{N}(p_{A B \vert XY})$ is upper bounded by a value that is always smaller than $1$. Furthermore, for fixed $N$, this bound approaches the bound for \eqref{eq:UpBoundP_H} when we take large values of $\kappa$ and is therefore violated by the quantum state and measurements described earlier.


\subsection*{Result 2: Quantum correlations which cannot be simulated if there is arbitrarily small OI}
\label{sec:ResultOI}


Consider the bipartite Bell experiment in which Alice and Bob have $M+1$ measurement options $x,y \in \{0,1,\ldots, M\}$, each of them with $2$ possible results, $a,b \in \{0,1\}$. Suppose that Alice and Bob share the following two-qubit entangled state:
\begin{equation}
\label{eq:StateLadderHardy}
    \ket{\phi} = \frac{1}{\sqrt{1+t^2}}(t \ket{00} - \ket{11}),
\end{equation}
where $t\in[0,1]$ is the value that maximises
\begin{equation}
    \max_{0 \leq t \leq 1} \frac{t^2}{1+t^2} \left(\frac{1 - t^{2M}}{1+t^{2M+1}} \right)^2.
\end{equation}
Alice's and Bob's measurements are of the form $A_{a|x} = |\pi_{a|x} \rangle \langle \pi_{a|x}|$ and $B_{b|y} = |\sigma_{b|y} \rangle \langle \sigma_{b|y}|$, with
\begin{eqnarray}
| \pi_{0|x} \rangle &=& \cos a_x | 0 \rangle + \sin a_x | 1 \rangle, \nonumber \forall x \in \{0,\ldots,M\},\\
| \pi_{1|x} \rangle &=& - \sin a_x | 0 \rangle + \cos a_x | 1 \rangle, \forall x \in \{0,\ldots,M\},
\end{eqnarray}
and 
\begin{eqnarray}
| \sigma_{0|y} \rangle &=& \cos b_y | 0 \rangle + \sin b_y | 1 \rangle, \forall y \in \{0,\ldots,M\}, \nonumber \\
| \sigma_{1|y} \rangle &=& - \sin b_y | 0 \rangle + \cos b_y | 1 \rangle, \forall y \in \{0,\ldots,M\},
\end{eqnarray}
with
\begin{equation}
    a_k = b_k = \arctan\left[ (-1)^k t^{k + 1/2} \right] \forall k \in \{0,\ldots,M\}.
\end{equation}
These state and measurements produce a correlation with the following properties:
\begin{subequations}\label{eq:LadderHardy-zeros}
\begin{align}
p(0,0|0,0) &= 0, \\
p(0,1|k,k-1) &= 0 \quad \forall k \in \{1,\ldots, M\},  \\
p(1,0|k-1,k) &= 0 \quad \forall k \in \{1,\ldots, M\},  
\end{align}
\end{subequations}
and correspond to the optimal implementation of the ``ladder'' version of Hardy's proof \cite{HardyPRL1993, ramanathan2021nosignalingproof, ZhaoQ2023}.

Let us define
\begin{align}
\label{eq:P_LadderhardyQuantum}
p_{H}^{M} \coloneqq p(0,0|M,M).
\end{align}

Result 2 can be formulated as follows: In any $\delta$-OD HV model that satisfies MI, PI and \eqref{eq:LadderHardy-zeros},
\begin{align}\label{eq:UpBoundP_H^M}
    p_{H}^{M} \leq \frac{\delta^{q}}{2},
\end{align}
where $q = \frac{M+1}{2}$. The proof is provided in the Supplementary Note 5. Therefore, if $\delta < 1$,  it follows that  $p_{H}^{M} < \frac{1}{2}$. In contrast, in quantum theory \cite{ramanathan2021nosignalingproof, ZhaoQ2023}, as $M$ approaches infinity, 
\begin{equation}
p_{H}^{M} \stackrel{M \rightarrow \infty}{\longrightarrow} \frac{1}{2}.
\end{equation}
Consequently, for any HV model with MI, PI and $\delta$-OD, with $\delta < 1$, there is $M$ such that quantum theory predicts a value for $p_{H}^{M}$ that cannot be simulated. 
In addition, it can be proven (see Supplementary Note 6) that the set of correlations produced by HV with MI, PI and complete OD is the set of nonsignaling correlations.

The above proof is based on the assumption that Eqs.~\eqref{eq:LadderHardy-zeros} hold. In an actual experiment, instead of the zeros Eqs.~\eqref{eq:LadderHardy-zeros}, we will obtain small values. Once we have them, we can derive an optimal Bell-like inequality that will allows us to discard any HV model with $\delta$-OD for some $\delta < 1$.


\section*{Discussion}
\label{sec:conclusions}


The results presented have consequences both for foundations and applications in quantum information processing, communication and computation. For foundations, our results bring us closer to the solution of a problem proposed by Shimony \cite{Shimony:1993}[pp.~96, 124, 149] \cite{Shimony:1993a}[p.~66] and which can be formulated as follows: ``One of these three premises [MI, PI and OI] must be false and it is important to locate the false one'' \cite{Shimony:1993}[p.~96]. 
If we assume that only {\em one} is false (and that it is the same one for all non-local quantum correlations), then
\begin{itemize}
\item[I.] If the assumption that fails is MI, Result 1 shows that, MI has to fail completely because there are quantum correlations that can only be explained with complete MD. 
\item[II.] If the assumption that fails is PI, Result 1 shows that, PI has to fail completely because the same correlations used in [I] can only be explained with complete PD. 
\item[III.] If the assumption that fails is OI, Result 2 shows that OI has to fail completely because there are quantum correlations that can only be explained with complete OD.
\end{itemize}


\begin{figure*}
\centering
\includegraphics[width=1.5\columnwidth]{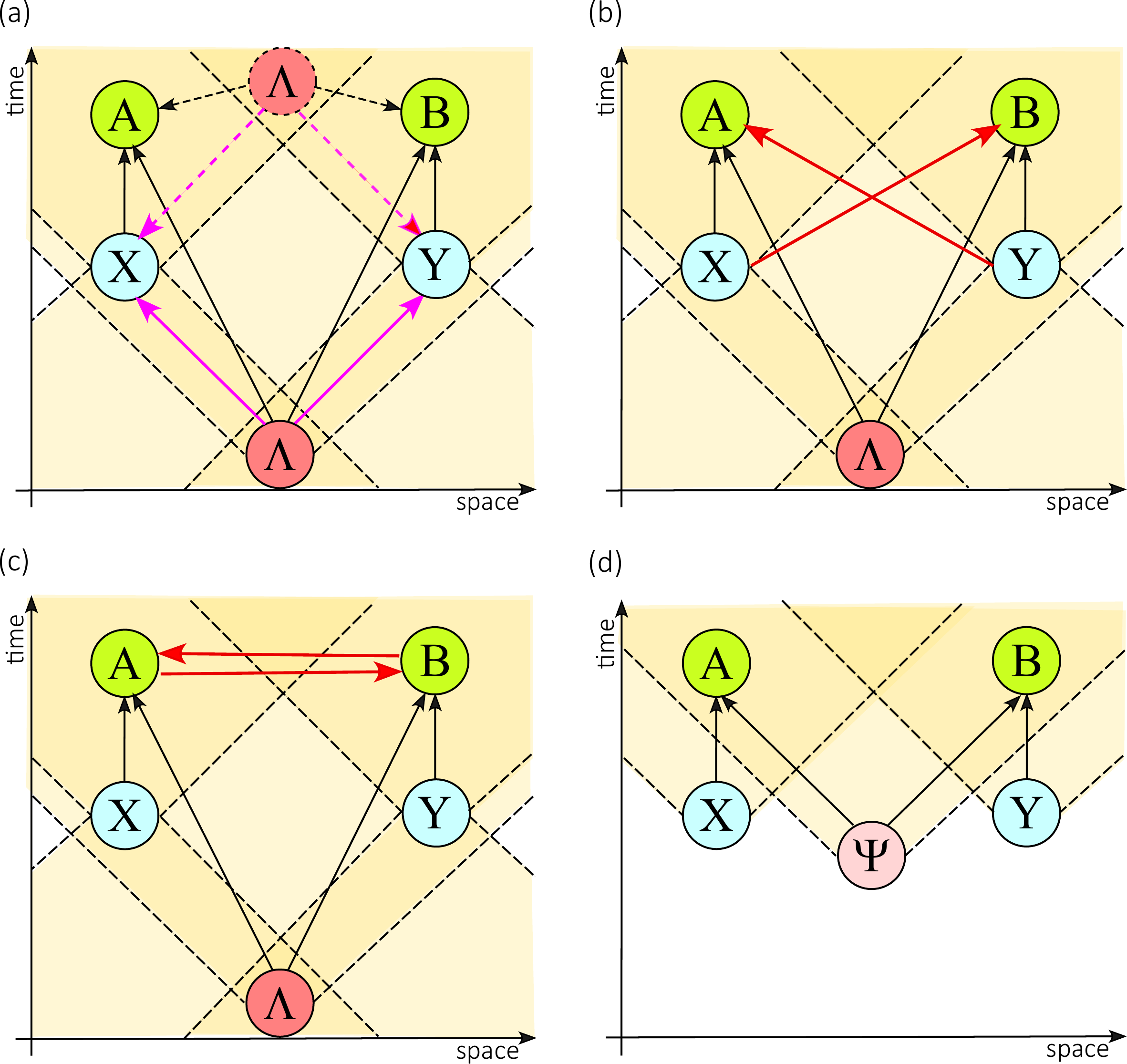}
\caption{Space-time diagrams of the causal influences needed, {\em in every round of the Bell-like test}, for each of the possible solutions to Shimony's problem in light of our results. In all diagrams, black arrows represent causal influences common to all possibilities. 
{\bf (a) Complete measurement dependence.} It can occur in, essentially, two ways. The first is with {\bf complete superdeterminism without retrocausality.} In this case, Alice and Bob do not have freedom of choice to choose the measurement settings, X and Y, respectively. Instead, the settings are determined by the  distribution of the HVs in the past light cones of X and Y, represented by the lower node $\Lambda$. The causal influences between the lower $\Lambda$ and X and Y are represented by violet continuous arrows. The second way complete MD can occur is with {\bf complete retrocausality with freedom of choice.} In this case, Alice and Bob have freedom to choose the ``nominal'' measurement settings X and Y, but the actual measurements are determined by the  distribution of the HVs in the future light cones of X and Y, represented by the upper node $\Lambda$. The causal influences between the upper $\Lambda$ and X and Y are represented by violet dashed arrows.
{\bf (b) Complete parameter dependence.} X is decided by Alice and Y is decided by Bob. However, X does not only influence Alice's measurement outcome, represented by A, but also Bob's measurement outcome, represented by B, which is outside the light cone of X. This superluminal influence (or ``action at a distance'') is represented by a red arrow. Similarly, Y does not only influence B, but also A.
{\bf (c) Complete outcome dependence.} X is decided by Alice and Y is decided by Bob. However, A and B are causally connected despite they are space-like separated. These superluminal influences are represented by blue arrows. 
{\bf (d) No hidden variables.} X is decided freely by Alice and Y is decided freely by Bob. A is causally connected only to the quantum state, represented by $\Psi$, and X. Similarly, B is causally connected only to $\Psi$ and Y.}
\label{fig:NoHV}
\end{figure*}


Each of these solutions to Shimony's problem requires extra causal influences which are not needed if HVs do not exist. These extra causal influences are shown with non-black colours in Figs.~\ref{fig:NoHV} (a), (b), (c), respectively. The causal influences if HVs do not exist are shown in Fig.~\ref{fig:NoHV} (d).

More generally, our results allow us to experimentally narrow down the possible explanations of Bell non-locality and the whole quantum theory, since they allow to experimentally excluding large subsets of HV models that are not excluded by previous experiments. Specifically, in principle, any $l$-MD, $(\varepsilon_A, \varepsilon_B)$-PD HV model with $l > 0$, $\varepsilon_A < 1$, $\varepsilon_B < 1$, satisfying OI can be experimentally excluded. Similarly, any $\delta$-OD HV model with $\delta < 1$ and satisfying MI and PI can be, in principle, experimentally excluded. Still, the experiments cannot exclude HV models with complete MD \cite{Brans:1988IJTP} or complete PD \cite{Bohm:1952PR, brask_bell_2017} or complete OD.

Result 1 extends the observation in \cite{putz_arbitrarily_2014} that there are quantum correlations that cannot be obtained from an $l$-MD HV model that satisfies OI and PI, for all values of $l>0$. In \cite{putz_arbitrarily_2014}, the difference between quantum theory and the models with OI, PI and arbitrarily small MI is so small that any relaxation of PI makes the difference to vanish. In this respect, Result 1 makes testable the impossibility of HV models with partial MD and PD. 

Result 2 is related to the observation in \cite{Ringbauer:2016SA} that there are quantum correlations that cannot be simulated with the assumptions of {\em causal models} (CM), MI, {\em causal parameter independence} (CPI) and the complete relaxation of {\em causal outcome independence} (COI). 
This observation is not made in the framework of the four assumptions of Bell's theorem (HV, MI, PI, OI) but in the framework of causal models \cite{wood_lesson_2015}. In general, neither HV and CM, nor PI and CPI, nor OI and COI, are equivalent.

In addition, Results 1 and 2 confirm that there is some interchangeability between MD and (PD+OD) \cite{conway_kochen_2011,blasiak_violations_2021}. However, our results go beyond that as they show that epsilon of each of MD and (PD+OD) is not enough: complete MD or complete PD or complete OD is needed. 

One reason why it is important to exclude HV models with partial (but not complete) MD is that these models have been proposed to explain quantum correlations 
\cite{colbeck_free_2012, brandao_realistic_2016, hall_complementary_2010, barretthow2011,hall_relaxed_2011, Hall:2015XXX, ramanathan2016randomness, hall_measurement-dependence_2020, ramanathan2021nosignalingproof}. In addition, partial MD or, more precisely partial human's free will, has been proposed in philosophy to resolve the conflict between the concept of an omniscient God and God's commandment not to commit sin \cite{Specker:1960D}.

One reason why it is important to exclude HV models with partial (but not complete) actions at a distance is that these models have been proposed to explain quantum correlations \cite{Brassard:1999PRL,Bacon:PRL2003,P2003PRA, pawlowski_non-local_2010,Ringbauer:2016SA,brask_bell_2017}. A second reason, which also applies to HV models with partial MD, is that excluding larger sets of HV models facilitates the discussion of the remaining models and, in particular, the discussion of the thermodynamics of the HV models \cite{PhysRevA.94.052127} that could not be discarded.

Our results are also of practical interest in quantum information processing, quantum communication and quantum computation. In the first place, for a general reason: the results show that quantum correlations do not only offer advantage with respect to {\em local} correlations, but also with respect to correlations assisted  by partial instantaneous actions at a distance or even assuming the existence of partial constraints to freedom of choice or partial retrocausal influences. This can make a big difference in quantum computational advantage. For example, when mapping quantum non-local correlations into the circuit model, the advantage of quantum theory {\em with respect to local HV theories} translates into a non-oracular quantum advantage \cite{bravyi_quantum_2018,Bravyi:2020NP}. Our results show that there is also advantage with respect to non-local correlations with partial MD, PD and OD. This may translate into new forms of quantum computational advantage.

Another reason why our results are of practical interest is device-independent (DI) quantum information processing \cite{liu_device-independent_2018,zhang_device-independent_2022}. DI protocols for random number generation \cite{pironio_random_2010}, quantum key distribution \cite{liu_device-independent_2018}, state tomography \cite{Yao_self} and self-testing of quantum devices \cite{MagniezSelf2006} achieve advantage allowing users to monitor the performance of their devices irrespective of noise, imperfections, and lack of knowledge regarding the inner workings - the users simply treat their devices as black boxes with classical inputs and outputs. An obstacle for practical DI protocols is the experimentally challenging requirement of a Bell test with: (I) quantum devices being isolated from each other, (II) with the inputs being chosen with uniform randomness and (III) with the detection loophole \cite{Pearle:1970PRD} closed. Experiments in different platforms \cite{Rowe:2001NAT,ansmann_violation_2009,Arute2019} allow for Bell tests with the detection loophole closed and high DI randomness generation rates. However, in these platforms, the quantum systems are very close, primarily to drive high entanglement generation rates via non-negligible coupling. The problem is that, precisely because the systems are close to one another, they can no longer be regarded as isolated in the sense needed for a Bell test. Sophisticated theoretical techniques have been devised to handle the issues of cross-talk and weak seeds separately. A Bell-like test allowing for arbitrary relaxation of MI and PI provides a simple and elegant solution to the problem of leakage of input information. A Bell-like test allowing  for simultaneous relaxation of MI, PI and OI would allow DI randomness generation tolerating weak seeds and cross-talk. We hope that our results will stimulate research in these directions.

\medskip 

\textbf{Acknowledgments.} The authors thank Paweł Horodecki and Pedro Lauand for discussions and acknowledge support from the Early Career Scheme (ECS) (Grant No.\ 27210620), the General Research Fund (GRF) (Grant No.\ 17211122), the Research Impact Fund (RIF) Grant No.\ R7035-21), the MCINN/AEI project ``New tools in quantum information and communication'' (Project No.\ PID2020-113738GB-I00), and the Digital Horizon Europe project \href{https://doi.org/10.3030/101070558}{``Foundations of quantum computational advantage'' (FoQaCiA)} (Grant agreement No.\ 101070558).


\onecolumngrid

\appendix


\section*{Supplementary Note 1. Stronger relaxation of MI}\label{app:noniid-MI}


In the main text we have considered the relaxation of MI in the following sense. We say that $p(a,b,x,y)$ is $l$-measurement dependent ($l$-MD) local if for all $x, y, \lambda$ it holds that $p(x,y|\lambda) \geq l > 0$. That is, in every round of the Bell experiment, the inputs $x, y$ are not perfectly determined by the hidden variable $\lambda$. This relaxation of MI is related to the task of randomness amplification, where the parties have access to a Santha-Vazirani source of partially random bits. The Santha-Vazirani source is characterised by the property that each bit taken from the source has an epsilonic amount of randomness even conditioned on all previous bits from the source. 

In this Supplementary Note, we consider a more general relaxation of the Measurement Independence assumption. Namely, we consider the scenario where in some rounds of the Bell experiment, the hidden variable completely determines the inputs (full measurement dependence) and in other rounds, the above $l$-MD relaxation holds. Such a relaxation is related to the randomness amplification of general Min-Entropy sources of partially random bits. That is, in this scenario, we consider that the parties in an $M$-run Bell experiment have access to a string of bits $x_1, \ldots, x_M, y_1, \ldots, y_M$ for which the only condition is that the maximum probability is bounded away from one, i.e., $p(x_1, \ldots, x_M, y_1, \ldots, y_M) \leq p_{\max} < 1$. There are two methods to approach this relaxation \cite{SBB23, R23}, with the second \cite{R23} being more suited for estimation and device-independent randomness amplification protocols. Here we illustrate the first approach from \cite{SBB23} without considering problems of estimation and device-independent security. 

We denote the rounds with full measurement dependence by the set of hidden variables $\lambda_{full-md}$ and the rounds with relaxed MI by the hidden variables $\lambda_{rel}$. We let $q$ denote the fraction of rounds with relaxed measurement dependence, i.e.
\begin{eqnarray}
 \int_{\lambda_{rel}} d \lambda p(\lambda) = q, \qquad \int_{\lambda_{full-md}} d \lambda p(\lambda) = 1 - q.
\end{eqnarray}
We now derive the bound for the value of a modified version of our Bell expression $\mathcal{I}$ in such hidden variable theories. 

Recall that in the ideal (noise-free) version of our Bell test, we have the Hardy zero constraints
\begin{subequations}\label{eq:HardyZeroSMI}
\begin{align}
& p(0, 1,a_{2}, b_{2}, \ldots,a_{N}, b_{N} \vert 0, 1)=\ldots=p(a_{1}, b_{1}, \ldots,a_{N-1}, b_{N-1},0,1 \vert 0, 1)=0,
\label{eq:HardyZero01SMI} \\
& p(1, 0, a_{2}, b_{2}, \ldots, a_{N}, b_{N} \vert 1, 0)=\ldots=p(a_{1}, b_{1}, \ldots, a_{N-1}, b_{N-1},1,0 \vert 1,0)=0,
\label{eq:HardyZero10SMI} \\				
& p(0, 0,a_{2}, b_{2}, \ldots,a_{N}, b_{N} \vert 1, 1)=\ldots=p(a_{1}, b_{1}, \ldots,a_{N-1}, b_{N-1},0,0 \vert 1, 1)=0,
\label{eq:HardyZero11SMI}
\end{align}
\end{subequations}
for all $a_1,\ldots,a_N,b_1,\ldots,b_N \in \{0,1\}$. Denote by $p_{Z,(0,1)}^{N}$, $p_{Z,(1,0)}^{N}$ and $p_{Z,(1,1)}^{N}$ the three sets of probabilities for inputs $(0,1), (1,0), (1,1)$ respectively in the \eqref{eq:HardyZeroSMI} above. Also recall that by definition
\begin{align}
\label{eq:P_hardyQuantumSM}
p_{H}^{N} \coloneqq \sum_{\substack{a_1, \ldots, a_N, b_1, \ldots, b_N \\ \left(a_1, b_1\right)=(0,0) \lor \ldots \lor \left(a_N, b_N\right)=(0,0)}} p\left(a_{1}, b_{1}, \ldots,a_{N}, b_{N} \vert 0, 0\right),
\end{align}
where $\lor$ is the logical OR. We have seen that in any $l$-MD and $(\varepsilon_A, \varepsilon_B)$-PD local model satisfying OI and Eqs.~\eqref{eq:HardyZero01SMI}, \eqref{eq:HardyZero10SMI}, and
\eqref{eq:HardyZero11SMI}, for all $l >0$, 
and all $N$, 
\begin{align}\label{eq:UpBoundP_H--SM}
p_{H}^{N}  \leq \varepsilon_{A}+\varepsilon_{B}-\varepsilon_{A} \varepsilon_{B}.
\end{align}
Therefore, if $\varepsilon_A < 1$ and $\varepsilon_B < 1$, then $p_{H}^{N} < 1$. In contrast, in quantum theory, as $N$ tends to infinity, 
\begin{equation}
p_{H}^{N} \stackrel{N \rightarrow \infty}{\longrightarrow} 1.
\end{equation}

Define $\mathcal{\tilde{I}}^{N}$ as
\begin{eqnarray}
    \mathcal{\tilde{I}}^{N} := &&l p_{X,Y}(0,0) \left[p_H^{N} - \left(\varepsilon_{A}+\varepsilon_{B}-\varepsilon_{A} \varepsilon_{B}\right)\right] - h p_{X,Y}(0,1) p_{Z,(0,1)}^{N} - h p_{X,Y}(1,0) p_{Z,(1,0)}^{N} - h p_{X,Y}(1,1) p_{Z,(1,1)}^{N} 
\end{eqnarray}
In the ideal noise-free version of the Bell test, we have that $p_{Z,(0,1)}^{N} = p_{Z,(1,0)}^{N} = p_{Z,(1,1)}^{N} = 0$ and the bound for the expression $\mathcal{\tilde{I}}^{N}$ in any $l$-MD and $(\varepsilon_A, \varepsilon_B)$-PD local model satisfying OI is $0$. Therefore, for the set of runs with relaxed measurement dependence $\lambda_{rel}$ we have that $\mathcal{\tilde{I}}^{N} \leq 0$. 

Now in the rounds with full measurement dependence, where the hidden variable $\lambda$ can completely determine the input distribution, the expression $\mathcal{\tilde{I}}^{N}$ has the algebraic bound $\mathcal{\tilde{I}}^{N} \leq l \left[1 - (\varepsilon_A + \varepsilon_B - \varepsilon_A \varepsilon_B) \right]$. 

We therefore obtain the bound for the strong-measurement dependent inequalities (still with $(\varepsilon_A, \varepsilon_B)$-Parameter Dependence in each round) as
\begin{eqnarray}
    \mathcal{\tilde{I}}^{N} &\leq& (1-q) \cdot l \cdot \left[1 - (\varepsilon_A + \varepsilon_B - \varepsilon_A \varepsilon_B) \right] + q \cdot 0 \nonumber \\
    &\leq& (1-  q) l \left[1 - (\varepsilon_A + \varepsilon_B - \varepsilon_A \varepsilon_B) \right].
\end{eqnarray}

Next, we evaluate the quantum bound for the expression $\mathcal{\tilde{I}}^{N}$. For this, we assume that all the input pairs appear with equal probabilities, i.e, $p_{X,Y}(x,y) = 1/4$ for all $x, y$. We then obtain the quantum bound 
\begin{eqnarray}
    \mathcal{\tilde{I}}^{N}_q = (l/4) \left[p_H^N -\left(\varepsilon_a + \varepsilon_B - \varepsilon_A \varepsilon_B \right) \right].
\end{eqnarray}
So that for the fraction of rounds with relaxed measurement dependence $q$ satisfying
\begin{eqnarray}
    &&(1-q)l\left[1 - (\varepsilon_A + \varepsilon_B - \varepsilon_A \varepsilon_B) \right] \leq (l/4) \left[p_H^N -\left(\varepsilon_a + \varepsilon_B - \varepsilon_A \varepsilon_B \right) \right] \nonumber \\
    \implies &&q \geq 1 - (1/4) \frac{\left[p_H^N -\left(\varepsilon_a + \varepsilon_B - \varepsilon_A \varepsilon_B \right) \right]}{\left[1 -\left(\varepsilon_a + \varepsilon_B - \varepsilon_A \varepsilon_B \right) \right]},
\end{eqnarray}
we have that quantum non-locality can be certified even under the strong measurement dependence relaxation. Therefore, in this approach, we require that full measurement dependence not occur in more than $(1/4)$-th fraction of runs in order to certify non-locality. An alternative approach following \cite{R23} which we pursue in forthcoming work shows how to improve this constraint so that quantum non-locality can be certified for any $q > 0$.


\section*{\texorpdfstring{Supplementary Note 1b (not included in the Nat.\ Comm.\ version). Other measures}{Supplementary Note 1b (not included in the published version). Other measures}}\label{app:OtherMeasures}


There are multiple measures of the correlation between the inputs $(X,Y)$ and the hidden variable $\Lambda$, as well as between Alice's input $X$ and Bob's output $B$. For instance, one may consider the mutual information $I(X,Y;\Lambda)$ and $I(X;B)$ as the natural measure of correlations, where in the latter case we evaluate the correlation conditioned on the values $y, \lambda$. There are two reasons why we do not use this alternative measure of correlations. First, the resulting optimization problems under constraints of arbitrary mutual information $I(X,Y; \Lambda)$ or $I(X;B)$ become quickly unwieldy and do not admit analytical solutions due to the heavily non-linear constraint. second, in the extreme case the relaxations converge. To elaborate, in the extreme situation of $\varepsilon_B = 1$, i.e., if for some $y$
we have
\begin{equation}\label{eq:dist-meas}
	\frac{1}{2} \sum_b \big\vert p(b|x=0, y, \lambda) - p(b|x=1, y, \lambda) \big\vert = 1,
\end{equation}
Bob is able to infer the value of $x$ from his local measurement outcome $b$ for that $y$. That is, the information gain by Bob is the maximal value of $H(X|y,\lambda)$ bits. To be clear, let us write the conditional entropy as
\begin{align}
	\label{eq:cond-entropy}
	H(X|B,y,\lambda) =& - \sum_{x,b} p(x,b|y,\lambda) \,\log \left[\frac{p(x,b|y,\lambda)}{p(b|y,\lambda)} \right]\nonumber \\
	=& - \sum_{x,b} p(x|y,\lambda) \,
    p(b|x,y,\lambda) \, \log \left[\frac{p(x|y,\lambda) \,p(b|x,y,\lambda) }{\sum_{x} p(x|y,\lambda) \,p(b|x,y,\lambda)} \right]. 
\end{align}
Now, from Eq.~\eqref{eq:dist-meas}, we see that, for any value of $b$, it holds that $p(b|x=0,y,\lambda) \neq 0 \implies p(b|x=1,y,\lambda) = 0$ and, similarly, $p(b|x=1,y,\lambda) \neq 0 \implies p(b|x=0,y,\lambda) = 0$. Therefore, in calculating the log term in Eq.~\eqref{eq:cond-entropy}, we see that the denominator $\sum_{x} p(x|y,\lambda) p(b|x,y,\lambda)$ is only non-zero for one value of $x$ for any given $b$. That is, we have that, for all $b$,
\begin{equation}
	\log \left[\frac{p(x|y,\lambda) \,p(b|x,y,\lambda) }{\sum_{x} p(x|y,\lambda) \,p(b|x,y,\lambda)} \right] =  \log \left[\frac{p(x|y,\lambda) \,p(b|x,y,\lambda) }{ p(x|y,\lambda) \,p(b|x,y,\lambda)} \right] = \log 1 = 0,
\end{equation}
giving $H(X|B,y,\lambda) = 0$. This therefore allows us to conclude that $I(X;B|y,\lambda) = H(X|y,\lambda) - H(X|B,y,\lambda) = H(X|y,\lambda)$ bits. In other words, Bob is able to achieve the maximum information gain about the value of random variable $X$ from his local measurement outcome $B$ for the given input $y$ for such $\lambda$. Furthermore, when the distance in Eq.~\eqref{eq:dist-meas} is less than $1$, i.e., when $\varepsilon_B < 1$, one can also deduce that $H(X|B,y,\lambda) > 0$ so that the mutual information is non-maximal.
Therefore, the extreme situation when $\varepsilon_B = 1$ exactly corresponds complete PD, i.e., there is a maximal correlation between Bob's measurement outcome and Alice's measurement setting. 

Another metric that could be used for the relaxation is the Kullback-Leibler (KL) divergence \cite{KL51}, given by
\begin{equation}
	D_{KL}(P||Q) = \sum_{k} P(k) \log \left[ \frac{P(k)}{Q(k)}\right],
\end{equation}
where $P$ and $Q$ may be thought of as the distributions of Bob's outcome $b$ conditioned on the different inputs $x=0$ and $x=1$ of Alice for some $y$ and $\lambda$. 
The total variation distance used in our paper, Eq.~\eqref{eq:dist-meas}, and the KL divergence are closely related. In fact, according to the Bretagnolle-Huber inequality \cite{BH78},
\begin{equation}
	\frac{1}{2}\sum_{k}|P(k) - Q(k)| \le \left(1 - e^{-D_{KL}(P||Q)} \right)^{1/2} \le 1 - \frac{e^{-D_{KL}(P||Q)}}{2}. 
\end{equation}
This relation implies that the total variation distance reaches its maximal value of $1$ only when the KL divergence diverges to infinity. Consequently, a full relaxation in terms of total variation distance necessarily entails a complete relaxation in KL divergence as well.

Therefore, the relaxation used in the paper, given by a distance measure, is natural and moreover amenable to analytical proof in the extreme case. 
Nevertheless, in an actual experimental scenario where extremal relaxations are not achievable, one may be interested in the optimization problems of MI, PI and OI relaxations in terms of mutual information or other measures.


\section*{\texorpdfstring{Supplementary Note 2. Proof of the upper bound for $p_{H}^{N}$ for \textit{l}-MD and $(\varepsilon_A, \varepsilon_B)$-PD HV models}{Supplementary Note 2. Proof of the upper bound for pH^N for l-MD and (εA, εB)-PD HV models}}\label{app:ProofUpperBound}

%
We are considering a Bell scenario with two parties, each of them choosing between two measurements with $2^N$ outcomes. We will label the outcomes as bit strings of size $N$. 
Here, we find an upper bound for $p_{H}^{N}$, defined as:
\begin{align}
\label{eq:P_hardyQuantum2}
p_{H}^{N} \coloneqq \sum_{\substack{a_1, \ldots, a_N, b_1, \ldots, b_N \\ \left(a_1, b_1\right)=(0,0) \lor \ldots \lor \left(a_N, b_N\right)=(0,0)}} p\left(a_{1}, b_{1}, \ldots,a_{N}, b_{N} \vert 0, 0\right),
\end{align}
over the set of $l$-MD and $(\varepsilon_A, \varepsilon_B)$-PD HV models and also over the hypotheses that the correlation satisfies the following:
\begin{subequations}\label{eq:HardyZeroSM}
\begin{align}
& p(0, 1,a_{2}, b_{2}, \ldots,a_{N}, b_{N} \vert 0, 1)=\ldots=p(a_{1}, b_{1}, \ldots,a_{N-1}, b_{N-1},0,1 \vert 0, 1)=0,
\label{eq:HardyZero01SM} \\
& p(1, 0, a_{2}, b_{2}, \ldots, a_{N}, b_{N} \vert 1, 0)=\ldots=p(a_{1}, b_{1}, \ldots, a_{N-1}, b_{N-1},1,0 \vert 1,0)=0,
\label{eq:HardyZero10SM} \\				
& p(0, 0,a_{2}, b_{2}, \ldots,a_{N}, b_{N} \vert 1, 1)=\ldots=p(a_{1}, b_{1}, \ldots,a_{N-1}, b_{N-1},0,0 \vert 1, 1)=0.
\label{eq:HardyZero11SM}
\end{align}
\end{subequations}
We show that, if $l>0$, $\varepsilon_A <1$, and $\varepsilon_B <1$, then this upper bound is violated by quantum theory.

For clarity, we will add subscripts to the probability distributions, so $p_{C}$ will represent the probability distribution of a random variable $C$, and $p_{C|D}$ will represent the conditional probability distribution of the variable $C$ given the variable $D$. In addition, $C$ or $D$ can be joint variables, as in $p_{AB \vert XY}$. 

An $l$-MD and $(\varepsilon_A, \varepsilon_B)$-PD HV model is a set of probability distributions $p_{A B \vert XY}$ that can be decomposed as follows:
\begin{align}\label{eq:OI_Apendice}
 p_{A B \vert XY} (\left(a_{1}, b_{1}\right), \ldots,\left(a_{N}, b_{N}\right) \vert x, y ) &=  \sum_{\lambda} p_{\Lambda|XY}(\lambda|x,y)p_{AB \vert XY \Lambda}(\left(a_{1}, b_{1}\right), \ldots,\left(a_{N}, b_{N}\right) \vert x, y,\lambda) \notag \\
&= \sum_{\lambda} p_{\Lambda|XY}(\lambda|x,y)p_{A \vert XY \Lambda}(a_1 \ldots, a_{N} \vert x, y, \lambda) p_{B \vert XY \Lambda}(b_{1}, \ldots, b_{N} \vert x, y, \lambda),
\end{align}
where, by the $l$-MD condition,
\begin{equation}\label{eq:MDdefinition_Append}
p_{XY|\Lambda}(x,y|\lambda) \ge l,
\end{equation}
and, by the $(\varepsilon_A, \varepsilon_B)$-PD condition,
\begin{subequations}\label{eq:PD_Append}
\begin{align}
& \frac{1}{2} \sum_{a_1 \ldots, a_{N}} |p_{A \vert XY \Lambda}(a_1 \ldots, a_{N} \vert x, 0) - p_{A \vert XY \Lambda}(a_1 \ldots, a_{N} \vert x, 1)| \le \varepsilon_A, \label{eq:PD_epsilonA} \\
   &  \frac{1}{2} \sum_{b_{1}, \ldots, b_{N}} |p_{B \vert XY \Lambda}(b_{1}, \ldots, b_{N} \vert 0, y) - p_{B \vert XY \Lambda}(b_{1}, \ldots, b_{N} \vert 1, y)| \le \varepsilon_B. \label{eq:PD_epsilonB}
\end{align}
\end{subequations}
$p_{A B \vert XY \Lambda}$ satisfies Eqs.~\eqref{eq:HardyZeroSM}. Therefore,
\begin{subequations}\label{eq:HardyZeroLambda}
\begin{align}
&p_{A \vert XY\Lambda}\left(0,a_2 \ldots, a_{N} \vert 0, 1, \lambda\right) p_{B \vert XY\Lambda}\left(1, b_{2}, \ldots, b_{N} \vert 0, 1, \lambda\right) = \ldots \notag \\
&= p_{A \vert XY\Lambda}\left(a_{1}, \ldots, a_{N-1},0 \vert 0, 1, \lambda\right) p_{B \vert XY\Lambda}\left(b_{1}, \ldots, b_{N-1},1 \vert 0, 1, \lambda\right) = 0,
\label{eq:HardyZero01Lambda}
\end{align}
\begin{align}
&p_{A \vert XY\Lambda}\left(1,a_2 \ldots, a_{N} \vert 1, 0, \lambda\right) p_{B \vert XY\Lambda}\left(0, b_{2}, \ldots, b_{N} \vert 1, 0, \lambda\right) = \ldots \notag \\
&= p_{A \vert XY\Lambda}\left(a_{1}, \ldots, a_{N-1},1 \vert 1, 0, \lambda\right) p_{B \vert XY\Lambda}\left(b_{1}, \ldots, b_{N-1},0 \vert 1, 0, \lambda\right) = 0,
\label{eq:HardyZero10Lambda}
\end{align}
\begin{align}
&p_{A \vert XY\Lambda}\left(0,a_2 \ldots, a_{N} \vert 1, 1, \lambda\right) p_{B \vert XY\Lambda}\left(0, b_{2}, \ldots, b_{N} \vert 1, 1, \lambda\right) = \ldots \notag \\
&= p_{A \vert XY\Lambda}\left(a_{1}, \ldots, a_{N-1}, 0 \vert 1, 1, \lambda\right) p_{B \vert XY\Lambda}\left(b_{1}, \ldots, b_{N-1},0 \vert  1, 1, \lambda\right) = 0,
\label{eq:HardyZero11Lambda}
\end{align}
\end{subequations}
for all $a_1,\ldots,a_N,b_1,\ldots,b_N \in \{0,1\}$ and all $\lambda$.

In this way, by Eq.~\eqref{eq:HardyZero01Lambda}, there is $\alpha_{1}, \ldots, \alpha_{k} \subseteq\{1, \ldots, N\}$ such that
\begin{subequations}\label{eq:AlphaAlphaBar}
\begin{align}
& p_{A \vert XY\Lambda}\left(a_{1}, \ldots, a_{\alpha_{i}-1}, 0, a_{\alpha_{i}+1}, \ldots, a_{N} \vert 0, 1, \lambda\right)=0,
\\
& p_{B \vert XY\Lambda}\left(b_1, \ldots, b_{\bar{\alpha}_{j}-1}, 1, b_{\bar{\alpha}_{j}+1}, \ldots, b_{N} \vert 0,1, \lambda\right)=0,
\end{align}
\end{subequations}
for all $i \in\{1, \ldots, k\}$, $j \in\{1, \ldots, N-k\}$, $a_{1}, \ldots, a_{N}, b_{1}, \ldots, b_{N} \in\{0,1\}$, and all $\lambda$, where $\left\{\bar{\alpha}_{1}, \ldots, \bar{\alpha}_{N-k}\right\}=\{1, \ldots, N\} \backslash\left\{\alpha_{1}, \ldots, \alpha_{k}\right\}$.

A similar reasoning applies to the Eqs.~\eqref{eq:HardyZero10Lambda} and \eqref{eq:HardyZero11Lambda}. By Eq.~\eqref{eq:HardyZero10Lambda}, there is $\beta_{1}, \ldots, \beta_{k^{\prime}} \subseteq\{1, \ldots, N\}$ such that
\begin{subequations}\label{eq:BetaBetaBar}
\begin{align}
& p_{A \vert XY\Lambda}\left(a_{1}, \ldots, a_{\beta_{i}-1}, 1, a_{\beta_{i}+1}, \ldots, a_{N} \vert 1,0, \lambda\right)=0,
\\
& p_{B \vert XY\Lambda}\left(b_{1}, \ldots, b_{\bar{\beta}_{j}-1}, 0, b_{\bar{\beta}_{j}+1}, \ldots, b_{N} \vert 1,0, \lambda\right)=0,
\end{align}
\end{subequations}
for all $i \in\{1, \ldots, k^{\prime}\}$, $j \in\{1, \ldots, N-k^{\prime}\}$, $a_{1}, \ldots, a_{N}, b_{1}, \ldots, b_{N} \in\{0,1\}$ and all $\lambda$, where $\left\{\bar{\beta}_{1}, \ldots, \bar{\beta}_{N-k^{\prime}}\right\}=\{1, \ldots, N\} \backslash\left\{ \beta_{1}, \ldots, \beta_{k^{\prime}}\right\}$. Finally, by Eq.~\eqref{eq:HardyZero11Lambda}, there is $\gamma_{1}, \ldots, \gamma_{k^{\prime \prime}} \subseteq\{1, \ldots, N\}$ such that
\begin{subequations}\label{eq:GammaGammaBar}
\begin{align}
& p_{A \vert XY\Lambda}\left(a_{1}, \ldots, a_{\gamma_{i-1}}, 0, a_{\gamma_{i+1}}, \ldots, a_{N} \vert 1, 1, \lambda\right)=0, 
\\
& p_{B \vert XY\Lambda}\left(b_{1}, \ldots, b_{\bar{\gamma}_{j-1}}, 0, b_{\bar{\gamma}_{j+1}}, \ldots, b_{N} \vert 1,1, \lambda\right)=0,
\end{align}
\end{subequations}
for all $i \in\{1, \ldots, k^{\prime\prime}\}, j \in\{1, \ldots, N-k^{\prime\prime}\}, a_{1}, \ldots, a_{N}, b_{1}, \ldots, b_{N} \in\{0,1\}$ and all $\lambda$, where $\left\{\bar{\gamma}_{1}, \ldots, \bar{\gamma}_{N-k^{\prime \prime}}\right\}=\{1, \ldots, N\} \backslash\left\{\gamma_{1}, \ldots, \gamma_{k^{\prime\prime}}\right\}$. 

The following Lemma shows that the restriction of the model $p_{AB \vert XY \Lambda}$ to be a $(\varepsilon_A, \varepsilon_B)$-PD HV model implies a relation between the sets $\left\{\alpha_{1}, \ldots, \alpha_{k}\right\}$ and $\left\{\beta_{1}, \ldots, \beta_{k^{\prime}}\right\}$.

\begin{lem}\label{lem:AlphaBeta}
Let
$
p_{AB \vert XY \Lambda}(\left(a_{1}, b_{1}\right), \ldots,\left(a_{N}, b_{N}\right) \vert x, y,\lambda) = p_{A \vert XY}(a_1 \ldots, a_{N} \vert x, y, \lambda) p_{B \vert XY}(b_{1}, \ldots, b_{N} \vert x, y, \lambda)$
be a $(\varepsilon_A, \varepsilon_B)$-PD HV model, for $\varepsilon_A, \varepsilon_B < 1$. Let us suppose that $p_{AB \vert XY\Lambda}$ satisfies Eqs.~\eqref{eq:HardyZeroLambda} and let $\left\{\alpha_{1}, \ldots, \alpha_{k}\right\}$, $\left\{\beta_{1}, \ldots, \beta_{k^{\prime}}\right\}$, $\left\{\gamma_{1}, \ldots, \gamma_{k^{\prime \prime}}\right\} \subseteq\{1, \ldots, N\}$ be the sets defined in Eqs.~\eqref{eq:AlphaAlphaBar}, \eqref{eq:BetaBetaBar}, and \eqref{eq:GammaGammaBar}, respectively. Then, $\left\{\bar{\alpha}_{1}, \ldots \bar{\alpha}_{N-k}\right\} \subseteq\left\{\bar{\beta}_{1}, \ldots, \bar{\beta}_{N-k^{\prime}}\right\}$.
\end{lem}
\begin{proof}
Since $\left\{\bar{\beta}_{1}, \ldots, \bar{\beta}_{N-k^{\prime}}\right\}$ is the complementary set of the set $\left\{\beta_{1}, \ldots, \beta_{k^{\prime}}\right\}$, then $\left\{\beta_{1}, \ldots, \beta_{k^{\prime}}\right\} \cup \left\{\bar{\beta}_{1}, \ldots, \bar{\beta}_{N-k^{\prime}}\right\} = \{1, \ldots, N\}$ and, therefore, $\left\{\bar{\alpha}_{1}, \ldots \bar{\alpha}_{N-k}\right\} \subseteq \left\{\beta_{1}, \ldots, \beta_{k^{\prime}}\right\} \cup \left\{\bar{\beta}_{1}, \ldots, \bar{\beta}_{N-k^{\prime}}\right\}$. In this way, proving Lemma~\ref{lem:AlphaBeta} is equivalent to showing that the intersection of $\left\{\bar{\alpha}_{1}, \ldots \bar{\alpha}_{N-k}\right\}$ and $\left\{\beta_{1}, \ldots, \beta_{k^{\prime}}\right\}$ is empty. We will see that, if this is not the case, then Eqs.~\eqref{eq:AlphaAlphaBar}, \eqref{eq:BetaBetaBar}, and \eqref{eq:GammaGammaBar} would be in contradiction with the assumption of $(\varepsilon_A, \varepsilon_B)$-PD HV model [Eq.~\eqref{eq:PD_Append}].

By contradiction, let us suppose that the intersection of $\left\{\bar{\alpha}_{1}, \ldots, \bar{\alpha}_{N-k}\right\}$ and $\left\{\beta_{1}, \ldots, \beta_{k'} \right\}$ is not empty. Therefore, given $r \in\left\{\bar{\alpha}_{1}, \ldots, \bar{\alpha}_{N-k}\right\} \cap\left\{\beta_{1}, \ldots, \beta_{k'} \right\}$, as this two set are subsets of $\{1, \ldots, N\}$, then $r \in\{1, \ldots, N\}=\left\{\gamma_{1}, \ldots, \gamma_{k^{\prime \prime}}\left\}\cup\left\{\bar{\gamma}_{1}, \ldots, \bar{\gamma}_{N-k^{\prime \prime}}\right\}\right.\right.$. In this way, $r \in\left\{\gamma_{1}, \ldots, \gamma_{k^{\prime \prime}}\right\}$ or $r \in\left\{\bar{\gamma}_{1}, \ldots, \bar{\gamma}_{N-k^{\prime \prime}}\right\}$. Let us deal with these two situations separately.

If $r\in\left\{\gamma_{1}, \ldots, \gamma_{k''}\right\}$: As $r$ is also in $\left\{\beta_{1}, \ldots, \beta_{k^{\prime}}\right\}$, by Eq.~\eqref{eq:BetaBetaBar},
\begin{align}
&  p_{A \vert XY\Lambda}\left(a_{1}, \ldots, a_{r-1}, 1, a_{r+1}, \ldots, a_{N} \vert 1, 0, \lambda\right)=0 \quad \forall a_{1}, \ldots, \hat{a}_{r}, \ldots, a_{N},
\end{align}
where $a_{1}, \ldots, \hat{a}_r, \ldots, a_{N}$ is a short notation for $a_{1}, \ldots, a_{r-1}, a_{r+1}, \ldots, a_{N}$. However, $p_{A \vert XY\Lambda}\left(a_{1}, \ldots, a_{N}|1, 0, \lambda\right)$ is a probability distribution and, as such, needs to satisfy normalisation. Combining these two ingredients, we have,
\begin{align}\label{eq:Prob010eq1}
1 & =\sum_{a_{1}, \ldots a_N} p_{A \vert XY\Lambda}\left(a_{1}, \ldots, a_{N}|1, 0, \lambda\right)\notag \\
& = \sum_{a_{1}, \ldots, \hat{a}_r,  \ldots, a_N} \left( p_{A \vert XY\Lambda}\left(a_{1}, \ldots, a_{r-1}, 0, a_{r-1},\ldots, a_{N}|1, 0, \lambda\right) +  \overbrace{p_{A \vert XY\Lambda}\left(a_{1}, \ldots, a_{r-1}, 1, a_{r-1},\ldots, a_{N}|1, 0, \lambda\right)}^{0} \right)\notag\\
&= \sum_{a_{1}, \ldots, \hat{a}_r,  \ldots, a_N} p_{A \vert XY\Lambda}\left(a_{1}, \ldots, a_{r-1}, 0, a_{r-1},\ldots, a_{N}|1, 0, \lambda\right).
\end{align}
Furthermore, using that $r \in \left\{\gamma_{1}, \ldots, \gamma_{k''} \right\}$, by Eq.~\eqref{eq:GammaGammaBar},
\begin{align}
&  p_{A \vert XY\Lambda}\left(a_{1}, \ldots, a_{r-1}, 0, a_{r+1}, \ldots, a_{N}\vert 1, 1, \lambda \right)=0 \quad \forall a_{1}, \ldots, \hat{a}_{r}, \ldots, a_{N}.
\end{align}
Analogously, since $p_{A \vert XY\Lambda}\left(a_{1}, \ldots, a_{N}\vert 1, 1, \lambda\right)$ is a probability distribution, it also satisfies normalization. Therefore,
\begin{align}\label{eq:Prob111eq1}
& 1= \sum_{a_{1}, \ldots, a_N} p_{A \vert XY\Lambda}\left(a_{1}, \ldots, a_{N}\vert 1, 1, \lambda\right) \notag\\
& = \sum_{a_{1}, \ldots, \hat{a}_r, \ldots, a_{N}} \left(\overbrace{p_{A \vert XY\Lambda}\left(a_{1}, \ldots, a_{r-1},0, a_{r+1}, \ldots, a_{N} \vert 1, 1, \lambda\right)}^{0} + p_{A \vert XY\Lambda}\left(a_{1}, \ldots, a_{r-1},1, a_{r+1}, \ldots, a_{N} \vert 1, 1, \lambda\right) \right)\notag\\
&= \sum_{a_{1}, \ldots, \hat{a}_r, \ldots, a_{N}} p_{A \vert XY\Lambda}\left(a_{1}, \ldots, a_{r-1},1, a_{r+1}, \ldots, a_{N} \vert 1, 1, \lambda\right).
\end{align}
We can combine Eqs.~\eqref{eq:Prob010eq1} and \eqref{eq:Prob111eq1} with the $(\varepsilon_A, \varepsilon_B)$-PD condition, Eq.~\eqref{eq:PD_epsilonA}, to obtain
\begin{align}
\varepsilon_{A} &\geqslant \frac{1}{2} \sum_{a_{1},\ldots, a_{N}}\left|p_{A \vert XY\Lambda}\left(a_{1}, \ldots, a_{N} \vert 1, 0, \lambda\right)-p_{A \vert XY\Lambda}\left(a_{1}, \ldots, a_{N} \vert 1, 1, \lambda\right)\right| \notag\\
& = \frac{1}{2}\sum_{a_{1}, \ldots, \hat{a}_{r}, \ldots, a_{N}} \vert p_{A \vert XY\Lambda}\left(a_{1}, \ldots, a_{r-1}, 0, a_{r}, \ldots, a_{N} \vert 1, 0, \lambda\right)-\overbrace{p_{A \vert XY\Lambda} (a_{1}, \ldots, a_{r-1}, 0, a_{r}, \ldots, a_{N}\vert 1, 1, \lambda)}^0 \vert \notag\\
& + \frac{1}{2}\sum_{a_{1}, \ldots, \hat{a}_{r}, \ldots, a_{N}} \vert \overbrace{p_{A \vert XY\Lambda}\left(a_{1}, \ldots, a_{r-1}, 1, a_{r}, \ldots, a_{N} \vert 1, 0, \lambda\right)}^0 - p_{A \vert XY\Lambda}\left(a_{1}, \ldots, a_{r-1}, 1, a_{r}, \ldots, a_{N}\vert 1, 1, \lambda\right) \vert \notag\\
& = \frac{1}{2}\sum_{a_{1}, \ldots, \hat{a}_{r}, \ldots, a_{N}} p_{A \vert XY\Lambda}\left(a_{1}, \ldots, a_{r-1}, 0, a_{r}, \ldots, a_{N} \vert 1, 0, \lambda\right) \notag\\
& + \frac{1}{2}\sum_{a_{1}, \ldots, \hat{a}_{r}, \ldots, a_{N}} p_{A \vert XY\Lambda}\left(a_{1}, \ldots, a_{r-1}, 1, a_{r}, \ldots, a_{N}\vert 1, 1, \lambda\right)  \notag\\
& = 1.
\end{align}
Therefore,
$ \varepsilon_{A} \geqslant 1$, which contradicts the assumption that $\varepsilon_{A} < 1$. In this way, if $ \varepsilon_{A} < 1$, $r$ cannot be in $\left\{\beta_{1}, \ldots, \beta_{k^{\prime}}\right\}$ and $\left\{\gamma_{1}, \ldots, \gamma_{k''} \right\}$ at the same time. 

Let us deal with the second case, where $r\in \left\{\bar{\gamma}_{1}, \ldots, \bar{\gamma}_{N-k''}\right\}$. As we will see, the arguments to reach a contradiction are completely analogous to the ones used in the first case and the conclusion reached will be that, if $\varepsilon_{B} < 1$, then $r$ cannot belong to the intersection of the sets $\left\{\bar{\alpha}_{1}, \ldots, \bar{\alpha}_{N-k}\right\}$ and $\left\{\bar{\gamma}_{1}, \ldots, \bar{\gamma}_{N-k''}\right\}$.

If $r\in \left\{\bar{\gamma}_{1}, \ldots, \bar{\gamma}_{N-k''}\right\}$: As $r$ is in $\left\{\bar{\alpha}_{1}, \ldots, \bar{\alpha}_{N-k}\right\}$, by Eq.~\eqref{eq:AlphaAlphaBar},
\begin{align}
&  p_{B \vert XY\Lambda}\left(b_{1}, \ldots, b_{r-1}, 1, b_{r+1}, \ldots, b_N \vert 0,1,\lambda\right)=0 \;\;\;\;\forall  b_{1}, \ldots, \hat{b}_{r}, \ldots, b_{N}.
\end{align}
Therefore, by using the normalization relation for $p_{B \vert XY\Lambda}\left(b_{1}, \ldots, b_{N} | 0,1,\lambda\right)$,
\begin{align}\label{eq:ProbB010eq1}
1 & = \sum_{b_{1}, \ldots, b_{N}} p_{B \vert XY\Lambda}\left(b_{1}, \ldots, b_{N} \vert  0,1,\lambda\right) \notag\\
& = \sum_{b_{1}, \ldots, \hat{b}_{r}, \ldots, b_{N}} \left(p_{B \vert XY\Lambda}\left(b_{1}, \ldots, b_{r-1}, 0, b_{r+1}, \ldots, b_{N} \vert  0,1,\lambda\right) 
 + \overbrace{p_{B \vert XY\Lambda}\left(b_{1}, \ldots, b_{r-1}, 1, b_{r+1}, \ldots, b_{N} \vert 0,1,\lambda\right)}^{0} \right)\notag\\
&= \sum_{b_{1}, \ldots, \hat{b}_{r}, \ldots, b_{N}} p_{B \vert XY\Lambda}\left(b_{1}, \ldots, b_{r-1}, 0, b_{r+1}, \ldots, b_{N} \vert 0,1,\lambda\right).
\end{align}
Moreover, using now that $r \in \left\{\bar{\gamma}_{1}, \ldots, \bar{\gamma}_{N-k''}\right\}$, by Eq.~\eqref{eq:GammaGammaBar},
\begin{align}
&  p_{B \vert XY\Lambda}\left(b_{1}, \ldots, b_{r-1}, 0, b_{r+1}, \ldots, b_{N} \vert 1,1, \lambda\right)=0, \;\;\;\;\forall b_{1}, \ldots, \hat{b}_{r}, \ldots, b_{N}.
\end{align}
Using normalization,
\begin{align} \label{eq:ProbB111eq1}
1 &= \sum_{b_{1}, \ldots, b_N} p_{B \vert XY\Lambda}\left(b_{1}, \ldots, b_{N} \vert 1, 1, \lambda\right) \notag\\
& = \sum_{b_{1}, \ldots, \hat{b}_r, \ldots, b_{N}} \left( \overbrace{p_{B \vert XY\Lambda}\left(b_{1}, \ldots, b_{r-1},0, b_{r+1}, \ldots, b_{N} \vert 1, 1, \lambda\right)}^{0} + p_{B \vert XY\Lambda}\left(b_{1}, \ldots, b_{r-1},1, b_{r+1}, \ldots, b_{N} \vert 1, 1, \lambda\right) \right) \notag\\
& = \sum_{b_{1}, \ldots, \hat{b}_r, \ldots, b_{N}} p_{B \vert XY\Lambda}\left(b_{1}, \ldots, b_{r-1},1, b_{r+1}, \ldots, b_{N} \vert 1, 1, \lambda\right).
\end{align}
Combining Eqs.~\eqref{eq:ProbB010eq1} and \eqref{eq:ProbB111eq1} with the $(\varepsilon_A, \varepsilon_B)$-PD condition [Eq.~\eqref{eq:PD_epsilonB}],
\begin{align}
\varepsilon_{B} &\geqslant \frac{1}{2}\sum_{b_{1}, \ldots, b_N} |p_{B \vert XY\Lambda}\left(b_{1}, \ldots, b_{N} \vert 0 ,1, \lambda\right) 
- p_{B \vert XY\Lambda}\left(b_{1}, \ldots, b_{N} \vert 1, 1, \lambda\right) | \notag\\
& = \frac{1}{2} \sum_{b_{1}, \ldots, \hat{b}_r, \ldots, b_{N}} \vert p_{B \vert XY\Lambda}\left(b_{1}, \ldots, b_{r-1},0, b_{r+1}, \ldots, b_{N} \vert 0, 1, \lambda\right)  
 - \overbrace{p_{B \vert XY\Lambda}\left(b_{1}, \ldots, b_{r-1},0, b_{r+1}, \ldots, b_{N} \vert 1, 1, \lambda\right)}^{0} | \notag \\
& + \frac{1}{2} \sum_{b_{1}, \ldots, \hat{b}_r, \ldots, b_{N}} |\overbrace{p_{B \vert XY\Lambda}\left(b_{1}, \ldots, b_{r-1},1, b_{r+1}, \ldots, b_{N} \vert 0, 1, \lambda\right)}^{0} 
 - p_{B \vert XY\Lambda}\left(b_{1}, \ldots, b_{r-1},1, b_{r+1}, \ldots, b_{N} \vert 1, 1, \lambda\right)|\notag \\
& = \frac{1}{2} \sum_{b_{1}, \ldots, \hat{b}_r, \ldots, b_{N}} p_{B \vert XY\Lambda}\left(b_{1}, \ldots, b_{r-1},0, b_{r+1}, \ldots, b_{N} \vert 0, 1, \lambda\right) \notag \\
& + \frac{1}{2}\sum_{b_{1}, \ldots, \hat{b}_r, \ldots, b_{N}}  p_{B \vert XY\Lambda}\left(b_{1}, \ldots, b_{r-1},1, b_{r+1}, \ldots, b_{N} \vert 1, 1, \lambda\right)\notag \\
& = 1. 
\end{align}
Therefore, 
$\varepsilon_{B} \geq 1$, which contradicts the assumption that $\varepsilon_{B} < 1$. 

Therefore, if $ \varepsilon_{A}, \varepsilon_{B} < 1$, $r$ being in $\left\{\bar{\alpha}_{1}, \ldots, \bar{\alpha}_{N-k}\right\} \cap\left\{\beta_{1}, \ldots, \beta_{k^{\prime}} \right\}$ implies that $r$ is not in $\left\{\gamma_{1}, \ldots, \gamma_{k''}\right\} \cup \left\{\bar{\gamma}_{1}, \ldots, \bar{\gamma}_{N-k^{''}}\right\} = \{1, \ldots, N\}$, which is a contradiction.
\end{proof}
The following
proposition will be useful to refine the upper bound 
$p_{H}^N$. This proposition is a general fact about probability distributions.
\begin{prop}\label{prop:Probs} Let $\Gamma$ be a finite sample space, $p_{1}$ and $p_{2}$ two probability distributions on $\Gamma$, which are $\eta$-close in total variation distance, \textit{i.e.},
\begin{equation}
\frac{1}{2}\sum_{\gamma \in \Gamma}\left|p_{1}(\gamma)-p_{2}(\gamma)\right| \leqslant \eta.
\end{equation}
Let $\Delta \subseteq \Gamma$ be such that $p_{2}(\Delta)=0$. Then, $p_1(\Delta)$ is upper bounded by
\begin{equation}
 p_{1}(\Delta) \coloneqq \sum_{\gamma \in \Delta} p_{1}(\gamma) \leqslant \eta.   
\end{equation}
\end{prop}
\begin{proof}
First, we should note that
\begin{align}
\sum_{\gamma \in \Gamma}\left|p_{1}(\gamma)-p_{2}(\gamma)\right| & =\sum_{\gamma \in \Delta}|p_{1}(\gamma)-\overbrace{p_{2}(\gamma)}^{0}|+\sum_{\gamma \in \bar{\Delta}}\left|p_{1}(\gamma)-p_{2}(\gamma)\right| \nonumber \\
& =\sum_{\gamma \in \Delta} p_{1}(\gamma)+\sum_{\gamma \in \bar{\Delta}}\left|p_{2}(\gamma)-p_{1}(\gamma)\right|,
\end{align}
where $\bar{\Delta} \coloneqq \Gamma \backslash \Delta$. On the other hand,
\begin{equation}
\sum_{\gamma \in \bar{\Delta}}\left|p_{2}(\gamma)-p_{1}(\gamma)\right| \geq \sum_{\gamma \in \bar{\Delta}} p_{2}(\gamma)-\sum_{\gamma \in \bar{\Delta}} p_{1}(\gamma)  =1-\left(1-\sum_{\gamma \in \Delta} p_{1}(\gamma)\right) =\sum_{\gamma \in \Delta} p_{1}(\gamma).
\end{equation}
Therefore, 
\begin{equation}
2 \sum_{\gamma \in \Delta} p_{1}(\gamma) \leqslant \sum_{\gamma \in \Delta} p_{1}(\gamma)+\sum_{\gamma \in \bar{\Delta}}\left|p_{1}(\gamma)-p_{2}(\gamma)\right| =\sum_{\gamma \in \Gamma}\left|p_{1}(\gamma)-p_{2}(\gamma)\right| \leqslant 2\eta .
\end{equation}
In this way,
\begin{equation}
\sum_{\gamma \in \Delta} p_{1}(\gamma) \leqslant \eta.
\end{equation}
\end{proof}

At this point, we have everything we need to prove the upper bound for $p_{H}^{N}$ provided in the main text.

\begin{proof}[Proof of the upper bound for $p_{H}^{N}$]
Using the variable $\Lambda$, we can express $p_{H}^{N}$ as,
\begin{align}
&p_{H}^{N}  \coloneqq \sum_{\substack{a_1, \ldots, a_N, b_1, \ldots, b_N \\ \left(a_1, b_1\right)=(0,0) \lor \ldots \lor \left(a_N, b_N\right)=(0,0)}} p_{A B \vert XY}\left(\left(a_{1}, b_{1}\right), \ldots,\left(a_{N}, b_{N}\right) \vert 0,0\right) \notag\\
&= \sum_{\lambda} p_{\Lambda}(\lambda) \frac{p_{XY \vert \Lambda}(0,0 \vert \lambda)}{p_{XY}(0, 0)} \sum_{\substack{a_1, \ldots, a_N, b_1, \ldots, b_N \\ \left(a_1, b_1\right)=(0,0) \lor \ldots \lor \left(a_N, b_N\right)=(0,0)}} p_{A B \vert XY\Lambda}\left(\left(a_{1}, b_{1}\right), \ldots,\left(a_{N}, b_{N}\right) \vert 0,0, \lambda\right).
\end{align}
For each $\lambda$, we define
\begin{equation}\label{eq:Defp_Hn}
p_{H}^{N,\lambda} \coloneqq \sum_{\substack{a_1, \ldots, a_N, b_1, \ldots, b_N \\ \left(a_1, b_1\right)=(0,0) \lor \ldots \lor \left(a_N, b_N\right)=(0,0)}} p_{A B \vert XY\Lambda}\left(\left(a_{1}, b_{1}\right), \ldots,\left(a_{N}, b_{N}\right) \vert 0,0, \lambda\right).
\end{equation}
The connection of $p_{H}^{N}$ and $p_{H}^{N,\lambda}$ is given by
\begin{equation}\label{eq:P_hardyAndP_hardyLambda}
 p_{H}^{N} = \sum_{\lambda} p_{\Lambda}(\lambda) \frac{p_{XY \vert \Lambda}(0,0 \vert \lambda)}{p_{XY}(0, 0)} p_{H}^{N,\lambda}.
\end{equation}
We will first find an upper bound for $p_{H}^{N,\lambda}$. Then, we will replace this upper bound in Eq.~\eqref{eq:P_hardyAndP_hardyLambda} and the theorem will be proven. We start by rewriting $p_{H}^{N,\lambda}$ using the partition $ \{1,\ldots,N\} = \left\{\alpha_{1}, \ldots, \alpha_{k}\right\} \cup \left\{\bar{\alpha}_{1}, \ldots, \bar{\alpha}_{N-k}\right\}$.
\begin{align}
& p_{H}^{N,\lambda}  = \sum_{\substack{a_1, \ldots, a_N, b_1, \ldots, b_N \\ \left(a_1, b_1\right)=(0,0) \lor \ldots \lor \left(a_N, b_N\right)=(0,0)}}  p_{A \vert XY\Lambda}\left(a_{1}, \ldots, a_{N} \vert 0, 0, \lambda\right)p_{B \vert XY\Lambda}\left(b_{1}, \ldots, b_{N}|0,0,\lambda\right) \notag \\
& = \sum_{\substack{a_1, \ldots, a_N, b_1, \ldots, b_N  \\ \left(a_{\alpha_1}, b_{\alpha_1}\right)=(0,0) \lor \ldots \lor \left(a_{\alpha_k}, b_{\alpha_k}\right)=(0,0)}} p_{A \vert XY\Lambda}\left(a_{1}, \ldots, a_{N} \vert 0, 0, \lambda\right) p_{B \vert XY\Lambda}\left(b_{1}, \ldots, b_{N} \vert 0, 0, \lambda\right) \notag\\
&+ \sum_{\substack{a_{1}, \ldots, a_{N}, b_{1}, \ldots, b_{N} \\ \left(a_{\bar{\alpha}_{1}}, b_{\bar{\alpha}_{1}}\right)=(0,0) \vee \ldots \vee\left(a_{\bar{\alpha}_{N-k},}, b_{\bar{\alpha}_{N-k}}\right)=(0,0) \\ \left(a_{\alpha_{1}}, b_{\alpha_{1}}\right) \neq(0,0) \wedge \ldots \wedge\left(a_{\alpha_{k}}, b_{\alpha_{k}}\right) \neq(0,0)}} p_{A \vert XY\Lambda}\left(a_{1}, \ldots, a_{N} \vert 0, 0, \lambda\right)p_{B \vert XY\Lambda}\left(b_{1}, \ldots, b_{N}\vert 0, 0, \lambda\right). 
\end{align}
We observe that the sums depend on the pair $(a_i,b_i)$ being equal to $(0,0)$. It is challenging to separate the sum into $a_i$ and $b_i$ independently. We can, however, do that at the cost of providing only an upper bound for $p_{H}^{N,\lambda}$. In fact,
\begin{align}\label{eq:P_hardyN}
p_{H}^{N,\lambda} &\leq \left(\sum_{\substack{a_{1},\ldots,a_{N} \\ a_{\alpha_1}=0 \lor \ldots \lor a_{\alpha_k}=0}}
p_{A \vert XY\Lambda}\left(a_{1}, \ldots, a_{N} \vert 0, 0, \lambda\right) \right)
\left( \sum_{\substack{b_{1},\ldots,b_{N} \\ b_{\alpha_1}=0 \lor \ldots \lor b_{\alpha_k}=0 \\ b_{\bar{\alpha}_{1}} \neq 0 \wedge \ldots \wedge b_{\bar{\alpha}_{k}} \neq 0}}
p_{B \vert XY\Lambda}\left(b_{1}, \ldots, b_{N}\vert 0, 0, \lambda\right) \right)\notag\\
& + \left(\sum_{\substack{a_{1},\ldots,a_{N} \\ a_{\bar{\alpha}_1}=0 \lor \ldots \lor a_{\bar{\alpha}_{N-k}}=0 \\ a_{\alpha_{1}} \neq 0 \wedge \ldots \wedge a_{\alpha_{k}} \neq 0 }} p_{A \vert XY\Lambda}\left(a_{1}, \ldots, a_{N} \vert 0,0,\lambda\right) \right)
\left(\sum_{\substack{b_{1},\ldots,b_{N} \\ b_{\bar{\alpha}_1}=0 \lor \ldots \lor b_{\bar{\alpha}_{N-k}}=0 }}
p_{B \vert XY\Lambda}\left(b_{1}, \ldots, b_{N} \vert 0, 0, \lambda\right) \right)\notag\\
& + \left(\sum_{\substack{a_{1},\ldots,a_{N} \\ a_{\alpha_1}=0 \lor \ldots \lor a_{\alpha_k}=0}}
p_{A \vert XY\Lambda}\left(a_{1}, \ldots, a_{N} \vert 0, 0, \lambda\right) \right) \left(\sum_{\substack{b_{1},\ldots,b_{N} \\ b_{\bar{\alpha}_1}=0 \lor \ldots \lor b_{\bar{\alpha}_{N-k}}=0}}
p_{B \vert XY\Lambda}\left(b_{1}, \ldots, b_{N} \vert 0, 0, \lambda\right) \right).
\end{align}
To simplify the notation, we will denote two of the sums above by
\begin{subequations}
\begin{align}
   & \vartheta_{\lambda} \coloneqq  \sum_{\substack{a_{1},\ldots,a_{N} \\ a_{\alpha_1}=0 \lor \ldots \lor a_{\alpha_k}=0}}  p_{A \vert XY\Lambda}\left(a_{1}, \ldots, a_{N} \vert 0, 0, \lambda\right), \label{eq:DefTheta} \\
   & \omega_{\lambda} \coloneqq  \sum_{\substack{b_{1},\ldots,b_{N} \\ b_{\bar{\alpha}_1}=0 \lor \ldots \lor b_{\bar{\alpha}_{N-k}}=0 }}
p_{B \vert XY\Lambda}\left(b_{1}, \ldots, b_{N}\vert 0, 0, \lambda\right). \label{eq:DefOmega}
\end{align}
\end{subequations}
We will see that it is possible to find an upper bound for $p_{H}^{N,\lambda}$ in terms of $\vartheta_{\lambda}$ and $\omega_{\lambda}$. To do this, we first observe that, by normalization of $ p_{A \vert XY\Lambda}$,
\begin{align}
1 &= \sum_{a_{1}, \ldots, a_{N}} p_{A \vert XY\Lambda}\left(a_{1}, \ldots, a_{N} \vert  0,0,\lambda\right) \notag\\ 
& = \sum_{\substack{a_{1},\ldots,a_{N} \\ a_{\alpha_1}=0 \lor \ldots \lor a_{\alpha_k}=0}}  p_{A \vert XY\Lambda}\left(a_{1}, \ldots, a_{N} \vert 0, 0, \lambda\right) \notag\\
&+ \sum_{\substack{a_{1},\ldots,a_{N} \\ a_{\bar{\alpha}_1}=0 \lor \ldots \lor a_{\bar{\alpha}_{N-k}}=0 \\ a_{\alpha_{1}} \neq 0 \wedge \ldots \wedge a_{\alpha_{k}} \neq 0 }} p_{A \vert XY\Lambda}\left(a_{1}, \ldots, a_{N} \vert 0,0,\lambda\right) + p_{A \vert XY\Lambda}\left(1, \ldots, 1 \vert  0,0,\lambda\right).
\end{align}
Therefore,
\begin{align}\label{eq:UB_theta}
\sum_{\substack{a_{1},\ldots,a_{N} \\ a_{\bar{\alpha}_1}=0 \lor \ldots \lor a_{\bar{\alpha}_{N-k}}=0 \\ a_{\alpha_{1}} \neq 0 \wedge \ldots \wedge a_{\alpha_{k}} \neq 0 }} p_{A \vert XY\Lambda}\left(a_{1}, \ldots, a_{N} \vert 0,0,\lambda\right) &= 1 - \vartheta_{\lambda} - p_{A \vert XY\Lambda}\left(1, \ldots, 1 \vert  0,0,\lambda\right)\notag \\
& \le 1 - \vartheta_{\lambda}.
\end{align}

We will deal now with the sum of Bob's probabilities. Repeating the idea of using the normalization of the probability distribution
 \begin{align}
1 &= \sum_{b_{1}, \ldots, b_{N}} p_{B \vert XY\Lambda}\left(b_{1}, \ldots, b_{N} \vert  0, 0 \lambda\right) \notag\\ 
& = \sum_{\substack{b_{1},\ldots,b_{N} \\ b_{\bar{\alpha}_1}=0 \lor \ldots \lor b_{\bar{\alpha}_{N-k}}=0 }} p_{B \vert XY\Lambda}\left(b_{1}, \ldots, b_{N} \vert 0, 0 \lambda\right) \notag \\ 
&+ \sum_{\substack{b_{1},\ldots,b_{N} \\ b_{\alpha_1}=0 \lor \ldots \lor b_{\alpha_k}=0 \\ b_{\bar{\alpha}_{1}} \neq 0 \wedge \ldots \wedge b_{\bar{\alpha}_{k}} \neq 0 }}
p_{B \vert XY\Lambda}\left(b_{1}, \ldots, b_{N} \vert 0, 0 \lambda\right) + p_{B \vert XY\Lambda}\left(1, \ldots, 1  \vert 0, 0, \lambda\right).
\end{align}
 Therefore,
\begin{align}\label{eq:UB_omega}
\sum_{\substack{b_{1},\ldots,b_{N} \\ b_{\alpha_1}=0 \lor \ldots \lor b_{\alpha_k}=0 \\ b_{\bar{\alpha}_{1}} \neq 0 \wedge \ldots \wedge b_{\bar{\alpha}_{k}} \neq 0 }}
p_{B \vert XY\Lambda}\left(b_{1}, \ldots, b_{N} \vert 0, 0 \lambda\right) &= 1 - \omega_{\lambda} - p_{B \vert XY\Lambda}\left(1, \ldots, 1  \vert 0, 0,\lambda\right)\notag \\
& \le 1 - \omega_{\lambda}.
\end{align}
We can now return to $p_{H}^{N, \lambda}$ [Eq.~\eqref{eq:P_hardyN}]. Using the definitions of $\vartheta_{\lambda}$ and $\omega_{\lambda}$ and Eqs.~\eqref{eq:UB_theta}-\eqref{eq:UB_omega}, we obtain the following upper bound for $p_{H}^{N, \lambda}$:
\begin{align}\label{UB_pHN}
p_{H}^{N, \lambda} &\le (1 - \omega_{\lambda})\vartheta_{\lambda} + \omega_{\lambda}(1 - \vartheta_{\lambda}) + \omega_{\lambda}\vartheta_{\lambda} \notag\\
&= \vartheta_{\lambda} + \omega_{\lambda} - \omega_{\lambda}\vartheta_{\lambda}.
\end{align}

The next step is to find upper bounds for $\vartheta_{\lambda}$ and $\omega_{\lambda}$. First, we should remember that if $r \in \{\alpha_1, \ldots, \alpha_k\}$, by Eq.~\eqref{eq:AlphaAlphaBar},
\begin{equation}
p_{A \vert XY\Lambda}\left(a_{1}, \ldots, a_{r}, 0, a_{r}, \ldots, a_{N} \vert 0, 1, \lambda\right)=0\;\;\;\;\forall a_{1}, \ldots, \hat{a}_r, \ldots, a_{N}.
\label{a3355}
\end{equation}
Defining $\Delta$ = $\{(a_{1},\ldots, a_{N}) \in \{0,1\}^{N} |  a_{\alpha_1}=0 \lor \ldots \lor a_{\alpha_k}=0\}$, by Eq.~\eqref{a3355}, 
\begin{equation}
p_{A \vert XY\Lambda}\left(\Delta \vert 0, 1, \lambda\right) = \sum_{\substack{a_{1},\ldots,a_{N} \\ a_{\alpha_1}=0 \lor \ldots \lor a_{\alpha_k}=0}} p_{A \vert XY\Lambda}\left(a_{1},  \ldots, a_{N} \vert 0, 1, \lambda\right) = 0.
\end{equation}
On the other hand, by the $(\varepsilon_A, \varepsilon_B)$-PD condition [Eq.~\eqref{eq:PD_epsilonA}],
\begin{equation}
 \frac{1}{2}\sum_{a_{1},\ldots, a_{N}}\left|p_{A \vert XY\Lambda}\left(a_{1}, \ldots, a_{N} \vert 0,0,\lambda\right)-p_{A \vert XY\Lambda}\left(a_{1}, \ldots, a_{N} \vert 0,1,\lambda\right)\right| \leq \varepsilon_{A}.
\end{equation}
Applying Proposition  \ref{prop:Probs}, we find the following upper bound for $\vartheta_{\lambda}$:
\begin{equation}
\vartheta_{\lambda} = \sum_{\substack{a_{1},\ldots,a_{N} \\ a_{\alpha_1}=0 \lor \ldots \lor a_{\alpha_k}=0}}p_{A \vert XY\Lambda}\left(a_{1}, \ldots, a_{N} \vert 0,0,\lambda\right) \leq \varepsilon_{A}. 
\end{equation}
The strategy for finding an upper bound for $\omega_{\lambda}$ is similar. In fact, by Lemma \ref{lem:AlphaBeta}, we have $\left\{\bar{\alpha}_{1}, \ldots, \bar{\alpha}_{N-k}\right\} \subseteq\left\{\bar{\beta}_{1}, \ldots, \bar{\beta}_{N-k'}\right\}$. In this way, for all $r \in\left\{\bar{\alpha}_{1}, \ldots, \bar{\alpha}_{N-k}\right\} \subseteq\left\{\bar{\beta}_{1}, \ldots, \bar{\beta}_{N-k'}\right\}$, by Eq.~\eqref{eq:BetaBetaBar},
\begin{equation}
\label{a34}
p_{B \vert XY\Lambda}\left(b_{1}, \ldots, b_{r-1}, 0, b_{r+1}, \ldots, b_{N} \vert 1, 0, \lambda\right)=0\;\;\;\;\forall b_{1} \ldots, \hat{b}_{r}, \ldots, b_{N}.
\end{equation}
Let $\Delta$ = $\{(b_{1},\ldots, b_{N}) \in \{0,1\}^{N} |  b_{\bar{\alpha}_1}=0 \lor \ldots \lor b_{\bar{\alpha}_{N-k}}=0\}$. Then, by Eq.~\eqref{a34}, 
\begin{equation}\label{eq:Delta_sumBob}
p_{B \vert XY\Lambda}\left(\Delta \vert 1, 0, \lambda\right) = \sum_{\substack{b_{1},\ldots,b_{N} \\ b_{\bar{\alpha}_1}=0 \lor \ldots \lor b_{\bar{\alpha}_{N-k}}=0}} p_{B \vert XY\Lambda}\left(b_{1},  \ldots, b_{N}  \vert 1, 0, \lambda\right) = 0.
\end{equation}
By the $(\varepsilon_A, \varepsilon_B)$-PD condition [Eq.~\eqref{eq:PD_epsilonB}],
\begin{align}
&\frac{1}{2}\sum_{b_{1},\ldots,b_{N} } \left| p_{B \vert XY\Lambda}\left(b_{1}, \ldots, b_{N} \vert  0,0,\lambda\right) - p_{B \vert XY\Lambda}\left(b_{1}, \ldots, b_{N} \vert  1, 0, \lambda\right) \right| \leq \varepsilon_B.
\end{align}
Applying Proposition \ref{prop:Probs}, we find the following upper bound for $\omega_{\lambda}$:
\begin{align}
&\omega_{\lambda} = \sum_{\substack{b_{1},\ldots,b_{N} \\ b_{\bar{\alpha}_1}=0 \lor \ldots \lor b_{\bar{\alpha}_{N-k}}=0}} p_{B \vert XY\Lambda}\left(b_{1}, \ldots, b_{N} \vert a_{1}, \ldots, a_{N}, 0,0,\lambda\right) \leq \varepsilon_B.
\end{align}
Therefore, 
\begin{equation}
p_{H}^{N, \lambda} \le \vartheta_{\lambda} + \omega_{\lambda} - \omega_{\lambda}\vartheta_{\lambda}, 
\end{equation}
where $\vartheta_{\lambda} \in [0, \varepsilon_A]$ and $\omega_{\lambda} \in [0,\varepsilon_B]$. 
This way, it is simple to see that the following upper bound is valid for $p_{H}^{N, \lambda}$:
\begin{align}
p_{H}^{N, \lambda} & \leq \varepsilon_{A}+\varepsilon_{B} - \varepsilon_{A} \varepsilon_{B}. \label{a4444}
\end{align}
Since this upper bound does not depend on $\lambda$, substituting Eq.~\eqref{a4444} in Eq.~\eqref{eq:P_hardyAndP_hardyLambda}, we conclude the proof,
\begin{align}
p_{H}^{N} &= \sum_{\lambda} p_{\Lambda}(\lambda) \frac{p_{XY \vert \Lambda}(0,0 \vert \lambda)}{p_{XY}(0, 0)} p_{H}^{N,\lambda} \notag\\
& \le \frac{(\varepsilon_{A}+\varepsilon_{B} - \varepsilon_{A} \varepsilon_{B})}{p_{XY}(0, 0)} \sum_{\lambda} p_{\Lambda}(\lambda) p_{XY \vert \Lambda}(0,0 \vert \lambda)\notag\\
& = \varepsilon_{A}+\varepsilon_{B} - \varepsilon_{A} \varepsilon_{B}. 
\end{align}
\end{proof}

Moreover, this upper bound for $p_{H}^{N}$ is tight. This can be seen by considering the following correlation:
\begin{equation}
p_{AB|XY} \left(a_{1}, b_{1}, \ldots,a_{N}, b_{N} \vert x, y\right) = p_{A|XY}\left(a_{1}, \ldots, a_{N}\vert x, y\right) p_{B|XY}\left(b_{1}, \ldots, b_{N}\vert x, y\right),
\end{equation}
where
\begin{subequations}
\begin{align}
&p_{A|XY}\left(a_{1}, \ldots, a_{N}\vert 0, 0\right) = \varepsilon_{A} \delta(a_{1} \ldots a_{N}, 0 \ldots 0) + (1 - \varepsilon_A) \delta(a_{1}, \ldots, a_{N}, 0 \ldots 01), \\
&p_{A|XY}\left(a_{1}, \ldots, a_{N}\vert 0, 1\right) = \varepsilon_{A} \delta(a_{1} \ldots a_{N}, 1 \ldots 1) + (1 - \varepsilon_A) \delta(a_{1}, \ldots, a_{N}, 0 \ldots 01), \\
&p_{A|XY}\left(a_{1}, \ldots, a_{N}\vert 1, 0\right) = \varepsilon_{A} \delta(a_{1} \ldots a_{N}, 0 \ldots 0) + (1 - \varepsilon_A) \delta(a_{1}, \ldots, a_{N}, 1 \ldots 10), \\
&p_{A|XY}\left(a_{1}, \ldots, a_{N}\vert 1, 1\right) = \varepsilon_{A} \delta(a_{1} \ldots a_{N}, 1 \ldots 1) + (1 - \varepsilon_A) \delta(a_{1}, \ldots, a_{N}, 1 \ldots 10), \\
&p_{B|XY}\left(b_{1}, \ldots, b_{N}\vert 0, 0\right) = \varepsilon_{B} \delta(b_{1} \ldots b_{N}, 0 \ldots 0) + (1 - \varepsilon_{B}) \delta(b_{1}, \ldots, b_{N}, 1 \ldots 10), \\
&p_{B|XY}\left(b_{1}, \ldots, b_{N}\vert 0, 1\right) = \varepsilon_{B} \delta(b_{1} \ldots b_{N}, 0 \ldots 0) + (1 - \varepsilon_{B}) \delta(b_{1}, \ldots, b_{N}, 0 \ldots 01), \\
&p_{B|XY}\left(b_{1}, \ldots, b_{N}\vert 1, 0\right) = \varepsilon_{B} \delta(b_{1} \ldots b_{N}, 1 \ldots 1) + (1 - \varepsilon_{B}) \delta(b_{1}, \ldots, b_{N}, 1 \ldots 10), \\
&p_{B|XY}\left(b_{1}, \ldots, b_{N}\vert 1, 1\right) = \varepsilon_{B} \delta(b_{1} \ldots b_{N}, 1 \ldots 1) + (1 - \varepsilon_{B}) \delta(b_{1}, \ldots, b_{N}, 0 \ldots 01). 
\end{align}
\end{subequations}
It follows that $p_{AB|XY}$ is $(\varepsilon_A, \varepsilon_B)$-PD and satisfies Eqs.~\eqref{eq:HardyZeroSM}. Moreover, for this correlation, $p_{H}^{N} = \varepsilon_{A}+\varepsilon_{B} - \varepsilon_{A} \varepsilon_{B}$.


\section*{Supplementary Note 3. Pseudo telepathy games and arbitrary relaxation of MI and PI}\label{app:PT_Games}


In \cite{putz_measurement_2016}, a protocol is presented to transform a Bell inequality into a new inequality that attests to non-locality over MI relaxations. Additionally, if the classical value of the original inequality is zero and the quantum value is strictly bigger than zero, then the new inequality created is even capable of attesting non-locality over arbitrary MI relaxations.

An interesting question is to look for what properties an inequality needs to have to allow us to attest to quantum violation on the hypothesis of arbitrary PI relaxations. In Sec.~\ref{app:NecessaryCondition} we enter into this discussion. We consider the reiteration of a Bell inequality as a Bell non-local game. With this, we show that, if a non-local game has a nonsignaling wining probability $\omega_{NS}(G) =1$, then, to allow arbitrary PI relaxations, it must be a pseudo telepathy (PT) game, i.e., it must have a quantum winning probability $\omega_Q(G) = 1$.


Due to these stunning properties of PT games, a natural question is whether they are always able to attest Bell non-locality over arbitrary relaxations of MI and PI. In Sec.~\ref{app:MG_Game_PI}, we delve deeper into this problem, providing an example of a PT game, known as the Magic Square game, which can be simulated using fixed and not complete relaxation of PI, thus,
\begin{italiccenter}
Not all PT games can certify non-locality over arbitrary relaxation of MI and PI.
\end{italiccenter}
This leads to the conclusion that PT games that attest non-locality to arbitrary relaxations of MI and PI, as shown in Supplementary Note 2, offer an even stronger test of Bell non-locality than general PT games.
\subsection{Necessary condition}\label{app:NecessaryCondition}
\begin{defi}
A non-local game is a 6-tuple $G = (X, Y, A, B, \pi, V) $, where $X$ and $Y$ are the input sets for Alice and Bob, respectively, $A$ and $B$ the respective output sets, $\pi: \pi(X,Y) $ is the input distribution, and $V: V(A,B,X,Y) \in \{0,1\} $ is the game's winning condition function. The classical and quantum winning probabilities of the game are defined as follows:
\begin{subequations}
\begin{align}
& w_{C}(G) := \max_{p(a,b|x,y) \in \text{C}} \sum_{x,y,a,b} \pi(x,y) \cdot p(a,b|x,y) \cdot V(a,b,x,y), \\
   &  w_{Q}(G) := \max_{p(a,b|x,y) \in \text{Q}} \sum_{x,y,a,b} \pi(x,y) \cdot p(a,b|x,y) \cdot V(a,b,x,y) = 1,
\end{align}
\end{subequations}
where the maximisation is made on the local and quantum set of correlations, respectively.
\end{defi}
It is easily to show that, for any game $G$ with $\omega_{NS}(G) = 1$, there exists a deterministic correlation $p_{AB|XY}$ with $p_{AB|XY}(ab|xy) = p_{A|XY}(a|xy)p_{B|Y}(b|y)$, $p_{AB|XY}(ab|xy) \in \{0,1\}$, such that $\omega_{p}(G) = 1$ \cite{Bacon:PRL2003, brask_bell_2017}. Using this fact, we can prove the following proposition:
\begin{prop}
Consider a game G with $\omega_{NS}(G) = 1$. If $\omega_Q(G) < \omega_{NS}(G)$, then there exists a $\varepsilon_A$-PD HV model, for $\varepsilon_A < 1$, such that
\begin{equation}
\omega_{p}(G) = \omega_{Q}(G).
\end{equation}
\end{prop}
\begin{proof}
Let 
\begin{equation}
p^{S}_{AB|XY}(ab|xy) = p^{S}_{A|XY}(a|xy)p^{S}_{B|Y}(b|y)
\end{equation}
be the correlation given that achieves maximum performance in the game $G$. Let $p^{C}_{A|X}(a|x)$ be a arbitrary distribution and let 
\begin{equation}
p^{C}_{AB|XY}(ab|xy) = p^{C}_{A|X}(a|x)p^{S}_{B|Y}(b|y),
\end{equation}
and 
\begin{equation}
r = \frac{\omega_Q(G) - \omega_{p^C}(G)}{1 - \omega_{p^C}(G)}.
\end{equation}
It follows that $r<1$, as $\omega_Q(G) < \omega_{NS}(G) = 1$. Then, let 
\begin{align}
\tilde{p}_{AB|XY}(ab|xy) &= r p^{S}_{AB|XY}(ab|xy) + (1 - r) p^{C}_{AB|XY}(ab|xy) \notag\\
&= (rp^{S}_{A|XY}(a|xy) + (1 - r)p^{C}_{A|X}(a|x))p^{S}_{B|Y}(b|y) \notag \\
&= \tilde{p}_{A|XY}(a|xy)\tilde{p}_{B|Y}(b|y).
\end{align}
It follows that
\begin{align}
\frac{1}{2}\sum_{a} |\tilde{p}\left(a \vert  \lambda, x, y\right) - p\left(a \vert \lambda, x, y'\right)| 
&= \frac{1}{2}\sum_{a} | rp^{S}_{A|XY}(a|xy) + (1 - r)p^{C}_{A|X}(a|x) - rp^{S}_{A|XY}(a|xy') - (1 - r)p^{C}_{A|X}(a|x)| \notag \\
& = \frac{1}{2}\sum_{a} r| p^{S}_{A|XY}(a|xy) - p^{S}_{A|XY}(a|xy')| \notag \\
&\le r.
\end{align}
   
On the other hand, it follows that
\begin{align}
\omega_{\tilde{p}}(G)  &= r\omega_{p^S}(G) + (1-r)\omega_{p^C}(G) \notag\\
&= \frac{\omega_Q(G) - \omega_{p^C}(G)}{1 - \omega_{p^C}(G)} + \left(1 - \frac{\omega_Q(G) - \omega_{p^C}(G)}{1 - \omega_{p^C}(G)} \right)\omega_{p^C}(G) \notag\\
&= \frac{\omega_Q(G) - \omega_{p^C}(G)}{1 - \omega_{p^C}(G)} + \left(\frac{\omega_{p^C}(G) - \omega_{p^C}(G)\omega_Q(G)}{1 - \omega_{p^C}(G)} \right) \notag\\
&= \omega_Q(G).
\end{align}
 
Therefore, $\tilde{p}_{AB|XY}$ is a $\varepsilon_A$-PD HV model, for $\varepsilon_A= r < 1$ and with $\omega_{p}(G) = \omega_{Q}(G)$.
\end{proof}

\begin{cor}
A non-local game that attests to arbitrary PI relaxations needs to be a PT game.
\end{cor}
\begin{proof}
Let $p^{Q}_{AB|XY}(ab|xy)$ be a quantum correlation that is not on the boundary of the nonsignaling polytope. Therefore, there is a non-local game s.t.~$\omega_{p^Q}(G) < \omega_{NS}(G)=1.$ Then, by the last result, we can achieve the same performance in this test using a $\varepsilon_A$-PD HV model with $\varepsilon_A < 2$.
\end{proof}


\subsection{The magic square game and relaxations of PI}\label{app:MG_Game_PI}


The magic square game is a cooperative game between two players, Alice and Bob \cite{PhysRevLett.86.1911, PhysRevLett.87.010403,Aravind:2004AJP}. It can be seen as a non-local game, where Alice and Bob receive one trit as input and need to provide three bits as output. The game performance is given by
\begin{equation}
w(p_{AB | XY}) := \frac{1}{9} \sum_{x,y,a,b} p_{AB | XY}(a_0, a_1, a_2, b_0, b_1, b_2|x,y) \cdot V(a_0, a_1, a_2, b_0, b_1, b_2,x,y),
\end{equation}
where
\begin{equation}
V(a_0, a_1, a_2, b_0, b_1, b_2,x,y) = 
\begin{cases}
1 & \text{ if } a_y = b_x, \text{ } a_0\oplus a_1 \oplus a_2 = 0 \text{ and } b_0\oplus b_1 \oplus b_2 = 1,\\
0 & \text{ otherwise.}
\end{cases}
\end{equation}
It is well known that
\begin{subequations}
\begin{align}
& w_{C}  := \max_{p_{AB | XY} \in \mathcal{C}} w(p_{AB | XY}) = \frac{8}{9} < 1, \\
& w_{Q} := \max_{p_{AB | XY} \in \mathcal{Q}} w(p_{AB | XY}) = 1,
\end{align}
\end{subequations}
where $\mathcal{C}$ and $\mathcal{Q}$ denote the classical and quantum correlation sets, respectively. 

Let us consider the following correlation:
\begin{equation}
\tilde{p}_{AB | XY}(a_0, a_1, a_2, b_0, b_1, b_2|x,y) = \tilde{p}_{A|X}(a_0, a_1, a_2|x)\tilde{p}_{B|XY}(b_0, b_1, b_2|x,y),
\end{equation}
where
\begin{subequations}
\begin{align}
& \tilde{p}_{A|X}(a_0, a_1, a_2|x) = \delta(a_0 a_1 a_2, 110), \\
& \tilde{p}_{B|XY}(b_0, b_1, b_2|x,0) = \tilde{p}_{B|XY}(b_0, b_1, b_2|0) =  \delta(b_{0} b_{1} b_{2}, 111) \quad \forall x, \\
&\tilde{p}_{B|XY}(b_0, b_1, b_2|x,1) = \tilde{p}_{B|XY}(b_0, b_1, b_2|1) =  \delta(b_{0} b_{1} b_{2}, 111)  \quad \forall x, \\
&\tilde{p}_{B|XY}(b_0, b_1, b_2|0,2) =  \frac{1}{2}\delta(b_{0} b_{1} b_{2}, 001) + \frac{1}{2}\delta(b_{0} b_{1} b_{2}, 010),\\
&\tilde{p}_{B|XY}(b_0, b_1, b_2|1,2) =  \frac{1}{2}\delta(b_{0} b_{1} b_{2}, 001) + \frac{1}{2}\delta(b_{0} b_{1} b_{2}, 100),\\
&\tilde{p}_{B|XY}(b_0, b_1, b_2|2,2) =  \frac{1}{2}\delta(b_{0} b_{1} b_{2}, 100) + \frac{1}{2}\delta(b_{0} b_{1} b_{2}, 010).
\end{align}
\end{subequations}

It follows that
\begin{equation}
\frac{1}{2}\sum_{b_{0}, b_{1}, b_{2}}\left|\tilde{p}_{B \vert X Y}\left(b_{0}, b_{1}, b_{2} \vert x, y\right) - \tilde{p}_{B \vert X Y}\left(b_{0}, b_{1}, b_{2} \vert x^{\prime}, y\right)\right| \leqslant \frac{1}{2}
\end{equation}
$\forall x, x^{\prime}, y$. Therefore, $\tilde{p}_{AB | XY}(a_0, a_1, a_2, b_0, b_1, b_2|x,y)$ is a $(0,1/2)$-PD HV model. Moreover, we can easily check that $w(\tilde{p}_{AB | XY}) = 1$. Therefore, we cannot have arbitrary relaxation of PI for the magic square game since $\varepsilon_{B}=1/2<1$ is already enough to get the maximum for $w(P_{AB | XY})$.


\section*{Supplementary Note 4. Derivation of the Bell-like inequality for HV models with arbitrary (but not complete) MD and PD}
\label{app:C}


Here, we will prove an upper bound for the Bell-like inequality:
\begin{align}\label{eq:IkappaSM}
   I_{\kappa}^{N}(p_{A B \vert XY}) &\coloneqq \sum_{\substack{a_1, \ldots, a_N, b_1, \ldots, b_N \\ \left(a_1, b_1\right)=(0,0) \lor \ldots \lor \left(a_N, b_N\right)=(0,0)}} p_{A B \vert XY}\left(\left(a_{1}, b_{1}\right), \ldots,\left(a_{N}, b_{N}\right) \vert 0,0\right) \notag \\
   &- \kappa \sum_{\substack{a_1, \ldots, a_N, b_1, \ldots, b_N \\ \left(a_1, b_1\right)=(0,1) \lor \ldots \lor \left(a_N, b_N\right)=(0,1)}} p_{A B \vert XY}\left(\left(a_{1}, b_{1}\right), \ldots,\left(a_{N}, b_{N}\right) \vert 0,1\right) \notag \\
& - \kappa \sum_{\substack{a_1, \ldots, a_N, b_1, \ldots, b_N \\ \left(a_1, b_1\right)=(1,0) \lor \ldots \lor \left(a_N, b_N\right)=(1,0)}} p_{A B \vert XY}\left(\left(a_{1}, b_{1}\right), \ldots,\left(a_{N}, b_{N}\right) \vert 1,0\right)\notag\\
&- \kappa \sum_{\substack{a_1, \ldots, a_N, b_1, \ldots, b_N \\ \left(a_1, b_1\right)=(0,0) \lor \ldots \lor \left(a_N, b_N\right)=(0,0)}} p_{A B \vert XY}\left(\left(a_{1}, b_{1}\right), \ldots,\left(a_{N}, b_{N}\right) \vert 1,1\right).
\end{align}

For simplicity, we will assume that the input probability distribution is uniform, \textit{i.e.}, $p(x,y) = \frac{1}{4}$ for all $x,y$. Hence, by Bayes's theorem,
\begin{equation}\label{eq:BayesApend}
p_{\Lambda|XY}(\lambda|x,y) = \frac{p_{\Lambda}(\lambda)p_{XY|\Lambda}(x,y|\lambda)}{p_{XY}(x,y)} = 4p_{\Lambda}(\lambda)p_{XY|\Lambda}(x,y|\lambda).
\end{equation}
It is important to emphasise that the distribution of the inputs being uniform does not impose any restrictions on the amount of correlation between $\lambda$ and $x,y$. 

Evaluating the functional $I_{\kappa}^{N}$ over this distribution $p_{AB|XY}$ of Eq.~\eqref{eq:OI_Apendice} and using Eq.~\eqref{eq:BayesApend}, we have
\begin{align}
   I_{\kappa}^{N}(p_{A B \vert XY}) &= 4\sum_{\lambda} p_{\Lambda}(\lambda)p_{XY|\Lambda}(0,0|\lambda)\sum_{\substack{a_1, \ldots, a_N, b_1, \ldots, b_N \\ \left(a_1, b_1\right)=(0,0) \lor \ldots \lor \left(a_N, b_N\right)=(0,0)}} p_{A B \vert XY}\left(\left(a_{1}, b_{1}\right), \ldots,\left(a_{N}, b_{N}\right) \vert 0,0, \lambda \right) \notag \\
   &- 4\kappa \sum_{\lambda} p_{\Lambda}(\lambda)p_{XY|\Lambda}(0,1|\lambda)\sum_{\substack{a_1, \ldots, a_N, b_1, \ldots, b_N \\ \left(a_1, b_1\right)=(0,1) \lor \ldots \lor \left(a_N, b_N\right)=(0,1)}} p_{A B \vert XY}\left(\left(a_{1}, b_{1}\right), \ldots,\left(a_{N}, b_{N}\right) \vert 0,1, \lambda \right) \notag \\
& - 4\kappa \sum_{\lambda} p_{\Lambda}(\lambda)p_{XY|\Lambda}(1,0|\lambda) \sum_{\substack{a_1, \ldots, a_N, b_1, \ldots, b_N \\ \left(a_1, b_1\right)=(1,0) \lor \ldots \lor \left(a_N, b_N\right)=(1,0)}} p_{A B \vert XY}\left(\left(a_{1}, b_{1}\right), \ldots,\left(a_{N}, b_{N}\right) \vert 1,0, \lambda \right)\notag\\
&- 4\kappa \sum_{\lambda} p_{\Lambda}(\lambda)p_{XY|\Lambda}(1,1|\lambda) \sum_{\substack{a_1, \ldots, a_N, b_1, \ldots, b_N \\ \left(a_1, b_1\right)=(0,0) \lor \ldots \lor \left(a_N, b_N\right)=(0,0)}} p_{A B \vert XY}\left(\left(a_{1}, b_{1}\right), \ldots,\left(a_{N}, b_{N}\right) \vert 1,1, \lambda \right).
\end{align}
Applying the $l$-MD condition [Eq.~\eqref{eq:MDdefinition_Append}], we can establish $l$ as a lower bound of the probability distributions $p_{XY|\Lambda}(0,1|\lambda)$, $p_{XY|\Lambda}(0,1|\lambda)$, and $p_{XY|\Lambda}(0,1|\lambda)$. Consequently, we can derive an upper bound for $I_{\kappa}^{N}(p_{A B \vert XY})$ as follows:
\begin{align}\label{eq:IkappaIneq1}
   I_{\kappa}^{N}(p_{A B \vert XY}) &\leq 4\sum_{\lambda} p_{\Lambda}(\lambda) p_{XY|\Lambda}(0,0|\lambda)\sum_{\substack{a_1, \ldots, a_N, b_1, \ldots, b_N \\ \left(a_1, b_1\right)=(0,0) \lor \ldots \lor \left(a_N, b_N\right)=(0,0)}} p_{A B \vert XY}\left(\left(a_{1}, b_{1}\right), \ldots,\left(a_{N}, b_{N}\right) \vert 0,0, \lambda \right) \notag \\
   &- 4l\kappa \sum_{\lambda} p_{\Lambda}(\lambda)\sum_{\substack{a_1, \ldots, a_N, b_1, \ldots, b_N \\ \left(a_1, b_1\right)=(0,1) \lor \ldots \lor \left(a_N, b_N\right)=(0,1)}} p_{A B \vert XY}\left(\left(a_{1}, b_{1}\right), \ldots,\left(a_{N}, b_{N}\right) \vert 0,1 , \lambda \right) \notag \\
& - 4l\kappa \sum_{\lambda} p_{\Lambda}(\lambda) \sum_{\substack{a_1, \ldots, a_N, b_1, \ldots, b_N \\ \left(a_1, b_1\right)=(1,0) \lor \ldots \lor \left(a_N, b_N\right)=(1,0)}} p_{A B \vert XY}\left(\left(a_{1}, b_{1}\right), \ldots,\left(a_{N}, b_{N}\right) \vert 1,0, \lambda \right)\notag\\
&- 4l\kappa \sum_{\lambda} p_{\Lambda}(\lambda) \sum_{\substack{a_1, \ldots, a_N, b_1, \ldots, b_N \\ \left(a_1, b_1\right)=(0,0) \lor \ldots \lor \left(a_N, b_N\right)=(0,0)}} p_{A B \vert XY}\left(\left(a_{1}, b_{1}\right), \ldots,\left(a_{N}, b_{N}\right) \vert 1,1, \lambda \right).
\end{align}
To simplify the notation, given a $\lambda$, we will represent the following sums as $\delta_{0,1}^{\lambda}, \delta_{1,0}^{\lambda}, \delta_{1,1}^{\lambda}$:
\begin{subequations}\label{eq:HardyZeroLambdaGeneral}
\begin{align}
& \delta_{0,1}^{\lambda} \coloneqq \sum_{\substack{a_1, \ldots, a_N, b_1, \ldots, b_N \\ \left(a_1, b_1\right)=(0,1) \lor \ldots \lor \left(a_N, b_N\right)=(0,1)}} p_{A \vert XY\Lambda}\left(a_1 \ldots, a_{N} \vert 0, 1, \lambda\right) p_{B \vert XY\Lambda}\left(b_{1}, \ldots, b_{N} \vert 0, 1, \lambda\right),
\label{eq:HardyZero01LambdaGeneral} \\
& \delta_{1,0}^{\lambda} \coloneqq \sum_{\substack{a_1, \ldots, a_N, b_1, \ldots, b_N \\ \left(a_1, b_1\right)=(1,0) \lor \ldots \lor \left(a_N, b_N\right)=(1,0)}} p_{A \vert XY\Lambda}\left(a_1 \ldots, a_{N} \vert 1, 0, \lambda\right) p_{B \vert XY\Lambda}\left(b_{1}, \ldots, b_{N} \vert 1, 0, \lambda\right),
\label{eq:HardyZero10LambdaGeneral}\\
& \delta_{1,1}^{\lambda} \coloneqq \sum_{\substack{a_1, \ldots, a_N, b_1, \ldots, b_N \\ \left(a_1, b_1\right)=(0,0) \lor \ldots \lor \left(a_N, b_N\right)=(0,0)}} p_{A \vert XY\Lambda}\left(a_1 \ldots, a_{N} \vert 1, 1, \lambda\right) p_{B \vert XY\Lambda}\left(b_{1}, \ldots, b_{N} \vert 1, 1, \lambda\right).
\label{eq:HardyZero11LambdaGeneral}
\end{align}
\end{subequations}
Therefore, given $r \in \{1,\ldots,N\}$, as a consequence of Eq.~\eqref{eq:HardyZero01LambdaGeneral}, we obtain
\begin{equation}
\sum_{\substack{a_{1}, \ldots, \hat{a}_r,  \ldots, a_N \\ b_{1}, \ldots, \hat{b}_r,  \ldots, b_N}} p_{A \vert XY\Lambda}\left(a_{1}, \ldots, a_{r-1}, 0, a_{r-1},\ldots, a_{N}|0, 1, \lambda\right) p_{B \vert XY\Lambda}\left(b_{1}, \ldots, b_{r-1}, 1, b_{r-1},\ldots, b_{N}|0, 1, \lambda\right) \leq \delta_{01}^{\lambda}.
\end{equation}
This expression can be reformulated as
\begin{align}
&\left(\sum_{a_{1}, \ldots, \hat{a}_r,  \ldots, a_N} p_{A \vert XY\Lambda}\left(a_{1}, \ldots, a_{r-1}, 0, a_{r-1},\ldots, a_{N}|0, 1, \lambda\right) \right) \notag \\
&\times \left(\sum_{ b_{1}, \ldots, \hat{b}_r,  \ldots, b_N}p_{B \vert XY\Lambda}\left(b_{1}, \ldots, b_{r-1}, 1, b_{r-1},\ldots, b_{N}|0, 1, \lambda\right) \right) \leq \delta_{01}^{\lambda}.
\end{align}
Consequently, at least one of the two sums must be upper bounded by $\sqrt{\delta_{01}^{\lambda}}$.
Therefore, Eq.~\eqref{eq:HardyZero01LambdaGeneral} implies the existence of $\alpha^{\lambda}_{1}, \ldots, \alpha^{\lambda}_{k} \subseteq\{1, \ldots, N\}$ such that
\begin{subequations}\label{eq:AlphaAlphaBarGeneral}
\begin{align}
& \sum_{a_{1}, \ldots, \hat{a}_{\alpha^{\lambda}_{i}},  \ldots, a_N} p_{A \vert XY\Lambda}\left(a_{1}, \ldots, a_{\alpha^{\lambda}_{i}-1}, 0, a_{\alpha^{\lambda}_{i}+1}, \ldots, a_{N} \vert 0, 1, \lambda\right) \leq \sqrt{\delta_{0,1}^{\lambda}} \quad \forall i \in\{1, \ldots, k\},
\label{eq:AlphaApendGeneral}
\\
& \sum_{b_{1}, \ldots, \hat{b}_{\bar{\alpha}^{\lambda}_{j}},  \ldots, b_N} p_{B \vert XY\Lambda}\left(b_1, \ldots, b_{\bar{\alpha}^{\lambda}_{j}-1}, 1, b_{\bar{\alpha}^{\lambda}_{j}+1}, \ldots, b_{N} \vert 0,1, \lambda\right) \leq \sqrt{\delta_{0,1}^{\lambda}}, \quad \forall j \in\{1, \ldots, N-k\},
\label{eq:AlphaBarApendGeneral}
\end{align}
\end{subequations}
where $\left\{\bar{\alpha}^{\lambda}_{1}, \ldots, \bar{\alpha}^{\lambda}_{N-k}\right\}=\{1, \ldots, N\} \backslash\left\{\alpha^{\lambda}_{1}, \ldots, \alpha^{\lambda}_{k}\right\}$.

We can apply a similar reasoning to  Eqs.~\eqref{eq:HardyZero10LambdaGeneral} and \eqref{eq:HardyZero11LambdaGeneral}. Indeed, Eq.~\eqref{eq:HardyZero10LambdaGeneral} implies the existence of $\beta^{\lambda}_{1}, \ldots, \beta^{\lambda}_{k^{\prime}} \subseteq\{1, \ldots, N\}$ such that
\begin{subequations}\label{eq:BetaBetaBarGen}
\begin{align}
& \sum_{a_{1}, \ldots, \hat{a}_{\beta^{\lambda}_{i}},  \ldots, a_N} p_{A \vert XY\Lambda}\left(a_{1}, \ldots, a_{\beta^{\lambda}_{i}-1}, 1, a_{\beta^{\lambda}_{i}+1}, \ldots, a_{N} \vert 1,0, \lambda\right) \leq \sqrt{\delta_{1,0}^{\lambda}} \quad \forall i \in\{1, \ldots, k^{\prime}\},
\label{eq:BetaApendGeneral}
\\
&
\sum_{b_{1}, \ldots, \hat{b}_{\bar{\beta}^{\lambda}_{j}},  \ldots, b_N}p_{B \vert XY\Lambda}\left(b, \ldots, b_{\bar{\beta}^{\lambda}_{j}-1}, 0, b_{\bar{\beta}^{\lambda}_{j}+1}, \ldots, b_{N} \vert 1,0, \lambda\right) \leq \sqrt{\delta_{1,0}^{\lambda}} \quad \forall j \in\{1, \ldots, N-k^{\prime}\},
\label{eq:BetaBarApendGeneral}
\end{align}
\end{subequations}
where $\left\{\bar{\beta}^{\lambda}_{1}, \ldots, \bar{\beta}^{\lambda}_{N-k^{\prime}}\right\}=\{1, \ldots, N\} \backslash\left\{ \beta^{\lambda}_{1}, \ldots, \beta^{\lambda}_{k^{\prime}}\right\}$. Finally, Eq.~\eqref{eq:HardyZero11LambdaGeneral} implies the existence of $\gamma^{\lambda}_{1}, \ldots, \gamma^{\lambda}_{k^{\prime \prime}} \subseteq\{1, \ldots, N\}$ such that
\begin{subequations}\label{eq:GammaGammaBarGen}
\begin{align}
& \sum_{a_{1}, \ldots, \hat{a}_{\gamma^{\lambda}_{i}},  \ldots, a_N} p_{A \vert XY\Lambda}\left(a_{1}, \ldots, a_{\gamma^{\lambda}_{i}-1}, 0, a_{\gamma^{\lambda}_{i}+1}, \ldots, a_{N} \vert 1, 1, \lambda\right) \leq \sqrt{\delta_{1,1}^{\lambda}}, \quad \forall i \in\{1, \ldots, k^{\prime\prime}\},
\label{eq:GammaApendGeneral}
\\
&
\sum_{b_{1}, \ldots, \hat{b}_{\bar{\gamma}^{\lambda}_{j}},  \ldots, b_N}p_{B \vert XY\Lambda}\left(b, \ldots, b_{\bar{\gamma}^{\lambda}_{j}-1}, 0, b_{\bar{\gamma}^{\lambda}_{j}+1}, \ldots, b_{N} \vert 1,1, \lambda\right) \leq \sqrt{\delta_{1,1}^{\lambda}}, \quad \forall j \in\{1, \ldots, N-k^{\prime\prime}\},
\label{eq:GammaBarApendGeneral}
\end{align}
\end{subequations}
where $\left\{\bar{\gamma}^{\lambda}_{1}, \ldots, \bar{\gamma}^{\lambda}_{N-k^{\prime \prime}}\right\}=\{1, \ldots, N\} \backslash\left\{\gamma^{\lambda}_{1}, \ldots, \gamma^{\lambda}_{k^{\prime\prime}}\right\}$. 

Before proceeding, let us state a straightforward fact about non-negative numbers.
\begin{fact}\label{fact:PosNumbers}
If $a$ and $b$ are non-negative numbers, then $ \sqrt{a} + \sqrt{b} \le \sqrt{2(a + b)}$.
\end{fact}
\begin{proof}
In general, for any real numbers $a$ and $b$, it is easy to see that $2ab \le a^2 + b^2$. Therefore, 
\begin{equation}
(a + b)^2 = a^2 + 2ab + b^2 \le 2(a^2 + b^2).
\end{equation}
Specifically, if $a, b \ge 0$, we obtain
\begin{equation}
\left( \sqrt{a} + \sqrt{b} \right)^2 \le 2(a + b).
\end{equation}
\end{proof}

The following lemma is an extension of Lemma \ref{lem:AlphaBeta} that replaces the requirement of satisfying Eqs.~\eqref{eq:HardyZeroSM} by a weaker condition: $l\kappa(\delta_{01}^{\lambda} + \delta_{10}^{\lambda} + \delta_{11}^{\lambda}) < 1$.

\begin{lem}\label{lem:AlphaBetaGen}
Let 
$p_{AB \vert XY \Lambda}(\left(a_{1}, b_{1}\right), \ldots,\left(a_{N}, b_{N}\right) \vert x, y,\lambda) = p_{A \vert XY\Lambda}(a_1 \ldots, a_{N} \vert x, y, \lambda) p_{B \vert XY\Lambda}(b_{1}, \ldots, b_{N} \vert x, y, \lambda)$ be a $(\varepsilon_A, \varepsilon_B)$-PD HV model, for $\varepsilon_A, \varepsilon_B < 1$. Let $\delta_{01}^{\lambda}$, $\delta_{10}^{\lambda}$, $\delta_{11}^{\lambda}$ defined by Eqs.~\eqref{eq:HardyZeroLambdaGeneral} and $\left\{\alpha^{\lambda}_{1}, \ldots, \alpha^{\lambda}_{k}\right\}$, $\left\{\beta^{\lambda}_{1}, \ldots, \beta^{\lambda}_{k^{\prime}}\right\}$, $\left\{\gamma^{\lambda}_{1}, \ldots, \gamma^{\lambda}_{k^{\prime \prime}}\right\}$ the sets defined in Eqs.~\eqref{eq:AlphaAlphaBarGeneral}, \eqref{eq:BetaBetaBarGen}, and \eqref{eq:GammaGammaBarGen}, respectively. Moreover, suppose that $l\kappa(\delta_{01}^{\lambda} + \delta_{10}^{\lambda} + \delta_{11}^{\lambda}) < 1$ where $\kappa > N^2 /l(1 - \varepsilon)^2$, where $\varepsilon = \max\{\varepsilon_A, \varepsilon_B\}$. Then, the set $\left\{\bar{\alpha}^{\lambda}_{1}, \ldots \bar{\alpha}^{\lambda}_{N-k}\right\}$ must be a subset of $\left\{\bar{\beta}^{\lambda}_{1}, \ldots, \bar{\beta}^{\lambda}_{N-k^{\prime}}\right\}$.
\end{lem}
\begin{proof}
The proof of this lemma is an adaptation of the proof of Lemma \ref{lem:AlphaBeta}. Let us suppose, by contradiction, that the intersection of $\left\{\bar{\alpha}^{\lambda}_{1}, \ldots, \bar{\alpha}^{\lambda}_{N-k}\right\}$ and $\left\{\beta^{\lambda}_{1}, \ldots, \beta^{\lambda}_{k'} \right\}$ is non-empty. Given $r \in\left\{\bar{\alpha}^{\lambda}_{1}, \ldots, \bar{\alpha}^{\lambda}_{N-k}\right\} \cap\left\{\beta^{\lambda}_{1}, \ldots, \beta^{\lambda}_{k'} \right\}$, since these two set are subsets of $\{1, \ldots, N\}$, then $r \in\{1, \ldots, N\}=\left\{\gamma^{\lambda}_{1}, \ldots, \gamma^{\lambda}_{k^{\prime \prime}}\left\}\cup\left\{\bar{\gamma}^{\lambda}_{1}, \ldots, \bar{\gamma}^{\lambda}_{N-k^{\prime \prime}}\right\}\right.\right.$. Consequently, $r \in\left\{\gamma^{\lambda}_{1}, \ldots, \gamma^{\lambda}_{k^{\prime \prime}}\right\}$ or $r \in\left\{\bar{\gamma}^{\lambda}_{1}, \ldots, \bar{\gamma}^{\lambda}_{N-k^{\prime \prime}}\right\}$. We will now examine these two situations separately.

If $r\in\left\{\gamma^{\lambda}_{1}, \ldots, \gamma^{\lambda}_{k''}\right\}$: Since $r$ also belongs to $\left\{\beta^{\lambda}_{1}, \ldots, \beta^{\lambda}_{k^{\prime}}\right\}$, according to Eq.~\eqref{eq:BetaBetaBarGen},
\begin{align}
&  \sum_{a_{1}, \ldots, \hat{a}_r,  \ldots, a_N} p_{A \vert XY\Lambda}\left(a_{1}, \ldots, a_{r-1}, 1, a_{r+1}, \ldots, a_{N} \vert 1, 0, \lambda\right) \leq \sqrt{\delta_{10}^{\lambda}}.
\end{align}
In this way, due to the normalization of $p_{A \vert XY\Lambda}\left(a_{1}, \ldots, a_{N}|1, 0, \lambda\right)$,
\begin{align}\label{eq:Prob010eq1Gen}
1 & =\sum_{a_{1}, \ldots a_N} p_{A \vert XY\Lambda}\left(a_{1}, \ldots, a_{N}|1, 0, \lambda\right)\notag \\
& = \sum_{a_{1}, \ldots, \hat{a}_r,  \ldots, a_N} p_{A \vert XY\Lambda}\left(a_{1}, \ldots, a_{r-1}, 0, a_{r-1},\ldots, a_{N}|1, 0, \lambda\right)\notag\\
&+ \overbrace{\sum_{a_{1}, \ldots, \hat{a}_r,  \ldots, a_N} p_{A \vert XY\Lambda}\left(a_{1}, \ldots, a_{r-1}, 1, a_{r-1},\ldots, a_{N}|1, 0, \lambda\right)}^{\leq \sqrt{\delta_{10}^{\lambda}}} \notag\\
&\leq  \sum_{a_{1}, \ldots, \hat{a}_r,  \ldots, a_N} p_{A \vert XY\Lambda}\left(a_{1}, \ldots, a_{r-1}, 0, a_{r-1},\ldots, a_{N}|1, 0, \lambda\right) + \sqrt{\delta_{10}^{\lambda}}.
\end{align}
Furthermore, using that $r \in \left\{\gamma^{\lambda}_{1}, \ldots, \gamma^{\lambda}_{k''} \right\}$, by Eq.~\eqref{eq:GammaGammaBarGen},
\begin{align}
&  \sum_{a_{1}, \ldots, \hat{a}_r,  \ldots, a_N} p_{A \vert XY\Lambda} p_{A \vert XY\Lambda}\left(a_{1}, \ldots, a_{r-1}, 0, a_{r+1}, \ldots, a_{N}\vert 1, 1, \lambda \right) \leq \sqrt{\delta_{11}^{\lambda}}.
\end{align}
Similarly, since $p_{A \vert XY\Lambda}\left(a_{1}, \ldots, a_{N}\vert 1, 1, \lambda\right)$ is a probability distribution, it also satisfies the normalization constraint. Hence,
\begin{align}\label{eq:Prob111eq1Gen}
& 1= \sum_{a_{1}, \ldots, a_N} p_{A \vert XY\Lambda}\left(a_{1}, \ldots, a_{N}\vert 1, 1, \lambda\right) \notag\\
& = \overbrace{\sum_{a_{1}, \ldots, \hat{a}_r, \ldots, a_{N}} p_{A \vert XY\Lambda}\left(a_{1}, \ldots, a_{r-1},0, a_{r+1}, \ldots, a_{N} \vert 1, 1, \lambda\right)}^{\leq \sqrt{\delta_{11}^{\lambda}}}\notag\\
&+\sum_{a_{1}, \ldots, \hat{a}_r, \ldots, a_{N}} p_{A \vert XY\Lambda}\left(a_{1}, \ldots, a_{r-1},1, a_{r+1}, \ldots, a_{N} \vert 1, 1, \lambda\right)\notag\\
&\leq \sum_{a_{1}, \ldots, \hat{a}_r, \ldots, a_{N}} p_{A \vert XY\Lambda}\left(a_{1}, \ldots, a_{r-1},1, a_{r+1}, \ldots, a_{N} \vert 1, 1, \lambda\right) + \sqrt{\delta_{11}^{\lambda}}.
\end{align}
We can combine Eqs.~\eqref{eq:Prob010eq1Gen} and \eqref{eq:Prob111eq1Gen} with the $(\varepsilon_A, \varepsilon_B)$-PD condition [Eq.~\eqref{eq:PD_epsilonA}] to obtain
\begin{align}
\varepsilon_{A} &\geqslant \frac{1}{2} \sum_{a_{1},\ldots, a_{N}}\left|p_{A \vert XY\Lambda}\left(a_{1}, \ldots, a_{N} \vert 1, 0, \lambda\right)-p_{A \vert XY\Lambda}\left(a_{1}, \ldots, a_{N} \vert 1, 1, \lambda\right)\right| \notag\\
& \ge  \frac{1}{2}\sum_{a_{1}, \ldots, \hat{a}_{r}, \ldots, a_{N}} ( p_{A \vert XY\Lambda}\left(a_{1}, \ldots, a_{r-1}, 0, a_{r}, \ldots, a_{N} \vert 1, 0, \lambda\right) -p_{A \vert XY\Lambda} (a_{1}, \ldots, a_{r-1}, 0, a_{r}, \ldots, a_{N}\vert 1, 1, \lambda)) \notag\\
& + \frac{1}{2}\sum_{a_{1}, \ldots, \hat{a}_{r}, \ldots, a_{N}} (- p_{A \vert XY\Lambda}\left(a_{1}, \ldots, a_{r-1}, 1, a_{r}, \ldots, a_{N} \vert 1, 0, \lambda\right)+ p_{A \vert XY\Lambda}\left(a_{1}, \ldots, a_{r-1}, 1, a_{r}, \ldots, a_{N}\vert 1, 1, \lambda\right) ) \notag\\
& =  \frac{1}{2}\overbrace{\sum_{a_{1}, \ldots, \hat{a}_{r}, \ldots, a_{N}} \left( p_{A \vert XY\Lambda}\left(a_{1}, \ldots, a_{r-1}, 0, a_{r}, \ldots, a_{N} \vert 1, 0, \lambda\right) 
+ p_{A \vert XY\Lambda}\left(a_{1}, \ldots, a_{r-1}, 1, a_{r}, \ldots, a_{N}\vert 1, 1, \lambda\right) \right)}^{ \ge 2 - \sqrt{\delta_{10}^{\lambda}} - \sqrt{\delta_{11}^{\lambda}} } \notag\\
& + \frac{1}{2}\overbrace{\sum_{a_{1}, \ldots, \hat{a}_{r}, \ldots, a_{N}} (- p_{A \vert XY\Lambda}\left(a_{1}, \ldots, a_{r-1}, 1, a_{r}, \ldots, a_{N} \vert 1, 0, \lambda\right) - p_{A \vert XY\Lambda}\left(a_{1}, \ldots, a_{r-1}, 0, a_{r}, \ldots, a_{N}\vert 1, 1, \lambda\right) )}^{ \ge - \sqrt{\delta_{10}^{\lambda}} - \sqrt{\delta_{11}^{\lambda}} } \notag\\
& \ge 1 - \left(\sqrt{\delta_{10}^{\lambda}} + \sqrt{\delta_{11}^{\lambda}} \right).
\end{align}
In addition, by hypothesis, $l\kappa(\delta_{01}^{\lambda} + \delta_{10}^{\lambda} + \delta_{11}^{\lambda}) \le 1$, and, consequently,
\begin{equation}
\delta_{10}^{\lambda} + \delta_{11}^{\lambda} < \frac{1}{l\kappa}.
\end{equation}
Applying Fact \ref{fact:PosNumbers}, we obtain
\begin{equation}
 \sqrt{\delta_{10}^{\lambda}} + \sqrt{\delta_{11}^{\lambda}} \leq \sqrt{2(\delta_{10}^{\lambda} + \delta_{11}^{\lambda})} \le \sqrt{\frac{2}{l\kappa}}.
\end{equation}
Therefore, we can deduce that
\begin{equation}
\varepsilon_A \ge 1 -  \left( \sqrt{\delta_{10}^{\lambda}} + \sqrt{\delta_{11}^{\lambda}} \right) \ge 1 - \sqrt{\frac{2}{l\kappa}}.
\end{equation}
However, given that $\kappa > N^2 /l(1 - \varepsilon_A)^2 > 2/l(1 - \varepsilon_A)^2$, where we are considering the cases of interest in which $N \ge 2$, we conclude that
\begin{equation}
 \sqrt{\frac{2}{l\kappa}} < 1 - \varepsilon_A.
\end{equation}
Consequently,
\begin{equation}
\varepsilon_A \ge 1 - \sqrt{\frac{2}{l\kappa}} > 1 - (1 - \varepsilon_A) = \varepsilon_A.
\end{equation}
which is a contradiction. Therefore, if $ \varepsilon_{A} < 1$, $r$ cannot simultaneously belong to $\left\{\beta^{\lambda}_{1}, \ldots, \beta^{\lambda}_{k^{\prime}}\right\}$ and $\left\{\gamma^{\lambda}_{1}, \ldots, \gamma^{\lambda}_{k''} \right\}$. Using entirely analogous reasoning (see Lemma \ref{lem:AlphaBeta}), we can show that $r$ also cannot simultaneously belong to $\left\{\bar{\alpha}^{\lambda}_{1}, \ldots, \bar{\alpha}^{\lambda}_{N-k}\right\}$ and $\left\{\bar{\gamma}^{\lambda}_{1}, \ldots, \bar{\gamma}^{\lambda}_{N-k''} \right\}$. Thus, the intersection of $\left\{\bar{\alpha}^{\lambda}_{1}, \ldots, \bar{\alpha}^{\lambda}_{N-k}\right\}$ and $\left\{\beta^{\lambda}_{1}, \ldots, \beta^{\lambda}_{k'} \right\}$ is empty, and consequently, $\left\{\bar{\alpha}^{\lambda}_{1}, \ldots \bar{\alpha}^{\lambda}_{N-k}\right\} \subseteq\left\{\bar{\beta}^{\lambda}_{1}, \ldots, \bar{\beta}^{\lambda}_{N-k^{\prime}}\right\}$.
\end{proof}
We now continue by providing an extension of the Proposition \ref{prop:Probs}.
\begin{prop}\label{prop:ProbsGeneral} Let $\Gamma$ be a finite sample space, and let $p_{1}$ and $p_{2}$ be two probability distributions on $\Gamma$, which are $\eta$-close in terms of the total variation distance,
\begin{equation}
\frac{1}{2}\sum_{\gamma \in \gamma}^{\lambda}\left|p_{1}(\gamma)-p_{2}(\gamma)\right| \leqslant \eta.
\end{equation}
Let $\Delta \subseteq \Gamma$, such that $p_{2}(\Delta)=\zeta$. Then, 
%
\begin{equation}
 p_{1}(\Delta) \coloneqq \sum_{\gamma \in \Delta} p_{1}(\gamma) \leqslant \eta + \zeta. 
\end{equation}
\end{prop}
\begin{proof}
Let $\bar{\Delta} \coloneqq \Gamma \backslash \Delta$, then,
\begin{equation}
\sum_{\gamma \in \bar{\Delta}}\left|p_{2}(\gamma)-p_{1}(\gamma)\right| \geq \sum_{\gamma \in \bar{\Delta}} p_{2}(\gamma)-\sum_{\gamma \in \bar{\Delta}} p_{1}(\gamma) 
= (1-\zeta) - \left(1-\sum_{\gamma \in \Delta} p_{1}(\gamma)\right) =\sum_{\gamma \in \Delta} p_{1}(\gamma) -\zeta.
\end{equation}
Therefore, 
\begin{align}
2 \sum_{\gamma \in \Delta} p_{1}(\gamma) & = \sum_{\gamma \in \Delta} p_{1}(\gamma) - \sum_{\gamma \in \Delta} p_{2}(\gamma) + \zeta + \sum_{\gamma \in \Delta} p_{1}(\gamma) \notag \\
&\leqslant \sum_{\gamma \in \Delta} |p_{1}(\gamma) - p_{2}(\gamma)|+\sum_{\gamma \in \bar{\Delta}}\left|p_{1}(\gamma)-p_{2}(\gamma)\right| + 2\zeta \leqslant 2\eta  + 2\zeta.
\end{align}
In this way,
\begin{equation}
\sum_{\gamma \in \Delta} p_{1}(\gamma) \leqslant \eta  + \zeta.
\end{equation}
\end{proof}
We can now proceed to find an upper for $I_{\kappa}^{N}$. Recalling that, according to Eq.~\eqref{eq:IkappaIneq1},
\begin{align}\label{eq:IkappaIneq2}
   I_{\kappa}^{N}(p_{A B \vert XY}) &\leq 4\sum_{\lambda} p_{\Lambda}(\lambda)p_{XY|\Lambda}(0,0|\lambda)\sum_{\substack{a_1, \ldots, a_N, b_1, \ldots, b_N \\ \left(a_1, b_1\right)=(0,0) \lor \ldots \lor \left(a_N, b_N\right)=(0,0)}} p_{A B \vert XY \Lambda}\left(\left(a_{1}, b_{1}\right), \ldots,\left(a_{N}, b_{N}\right) \vert 0,0, \lambda \right) \notag \\
   &- 4l\kappa \sum_{\lambda} p_{\Lambda}(\lambda)\sum_{\substack{a_1, \ldots, a_N, b_1, \ldots, b_N \\ \left(a_1, b_1\right)=(0,1) \lor \ldots \lor \left(a_N, b_N\right)=(0,1)}} p_{A B \vert XY \Lambda}\left(\left(a_{1}, b_{1}\right), \ldots,\left(a_{N}, b_{N}\right) \vert 0,1,\lambda\right) \notag \\
& - 4l\kappa \sum_{\lambda} p_{\Lambda}(\lambda) \sum_{\substack{a_1, \ldots, a_N, b_1, \ldots, b_N \\ \left(a_1, b_1\right)=(1,0) \lor \ldots \lor \left(a_N, b_N\right)=(1,0)}} p_{A B \vert XY \Lambda}\left(\left(a_{1}, b_{1}\right), \ldots,\left(a_{N}, b_{N}\right) \vert 1,0, \lambda\right)\notag\\
&- 4l\kappa \sum_{\lambda} p_{\Lambda}(\lambda) \sum_{\substack{a_1, \ldots, a_N, b_1, \ldots, b_N \\ \left(a_1, b_1\right)=(0,0) \lor \ldots \lor \left(a_N, b_N\right)=(0,0)}} p_{A B \vert XY \Lambda}\left(\left(a_{1}, b_{1}\right), \ldots,\left(a_{N}, b_{N}\right) \vert 1,1, \lambda\right).
\end{align}
To simplify the notation, we define
\begin{equation}\label{eq:pHnGeneral}
p_{H}^{N,\lambda} \coloneqq \sum_{\substack{a_1, \ldots, a_N, b_1, \ldots, b_N \\ \left(a_1, b_1\right)=(0,0) \lor \ldots \lor \left(a_N, b_N\right)=(0,0)}} p_{A B \vert XY\Lambda}\left(\left(a_{1}, b_{1}\right), \ldots,\left(a_{N}, b_{N}\right) \vert 0,0, \lambda\right).
\end{equation}
Thus, combining Eq.~\eqref{eq:IkappaIneq2} with Eqs.~\eqref{eq:HardyZeroLambdaGeneral} and~\eqref{eq:pHnGeneral}, we obtain
\begin{equation}
I_{\kappa}^{N}(p_{A B \vert XY}) \le 4\sum_{\lambda} p_{\Lambda}(\lambda) p_{XY|\Lambda}(0,0|\lambda) p_{H}^{N,\lambda} - 4l\kappa \sum_{\lambda} p_{\Lambda}(\lambda)(\delta_{01}^{\lambda} + \delta_{10}^{\lambda} + \delta_{11}^{\lambda}).
\end{equation}
Let $\Lambda$ be the set of all values of the hidden variable $\lambda$ of our model. We then define the following partition: $\Lambda = \Lambda_{1} \cup \Lambda_{2}$, where $ \Lambda_{1} = \{ \lambda \in \Lambda | l\kappa(\delta_{01}^{\lambda} + \delta_{10}^{\lambda} + \delta_{11}^{\lambda}) \le 1 \}$ and $ \Lambda_{2} = \{ \lambda \in \Lambda | l\kappa(\delta_{01}^{\lambda} + \delta_{10}^{\lambda} + \delta_{11}^{\lambda}) > 1 \}$. We observe that
\begin{align}
I_{\kappa}^{N}(p_{A B \vert XY}) &\le 4\sum_{\lambda \in \Lambda_1} p_{\Lambda}(\lambda) \left(p_{XY|\Lambda}(0,0|\lambda) p_{H}^{N,\lambda} - l\kappa(\delta_{01}^{\lambda} + \delta_{10}^{\lambda} + \delta_{11}^{\lambda}) \right) \notag \\
&+ 4\sum_{\lambda\in \Lambda_2} p_{\Lambda}(\lambda) \left(p_{XY|\Lambda}(0,0|\lambda) p_{H}^{N,\lambda} - l\kappa(\delta_{01}^{\lambda} + \delta_{10}^{\lambda} + \delta_{11}^{\lambda}) \right) \notag\\
&\le 4\sum_{\lambda \in \Lambda_1} p_{\Lambda}(\lambda) \left(p_{XY|\Lambda}(0,0|\lambda) p_{H}^{N,\lambda} - l\kappa(\delta_{01}^{\lambda} + \delta_{10}^{\lambda} + \delta_{11}^{\lambda}) \right),
\end{align}
where the last inequality follows from the fact that $p_{XY|\Lambda}(0,0|\lambda) p_{H}^{N,\lambda} - l\kappa(\delta_{01}^{\lambda} + \delta_{10}^{\lambda} + \delta_{11}^{\lambda}) < 0$ for $\lambda \in \Lambda_2 $. Consequently, we can focus on the $\lambda$ values where  
$l\kappa(\delta_{01}^{\lambda} + \delta_{10}^{\lambda} + \delta_{11}^{\lambda}) \le 1$, and Lemma \ref{lem:AlphaBetaGen} applies for these $\lambda$ values.

As we can see, $p_{H}^{N,\lambda}$ defined in Eq.~\eqref{eq:pHnGeneral}, is equal to $p_{H}^{N,\lambda}$, defined in Eq.~\eqref{eq:Defp_Hn}. Moreover, following the same steps, we obtain the upper bound [Eq.~\eqref{UB_pHN}],
\begin{align}
p_{H}^{N, \lambda} = \vartheta_{\lambda} + \omega_{\lambda} - \omega_{\lambda}\vartheta_{\lambda},
\end{align}
where
\begin{subequations}
\begin{align}
& \vartheta_{\lambda} \coloneqq  \sum_{\substack{a_{1},\ldots,a_{N} \\ a_{\alpha^{\lambda}_1}=0 \lor \ldots \lor a_{\alpha^{\lambda}_k}=0}}  p_{A \vert XY\Lambda}\left(a_{1}, \ldots, a_{N} \vert 0, 0, \lambda\right), \label{eq:DefThetaGeneral} \\
&\omega_{\lambda} \coloneqq  \sum_{\substack{b_{1},\ldots,b_{N} \\ b_{\bar{\alpha}^{\lambda}_1}=0 \lor \ldots \lor b_{\bar{\alpha}^{\lambda}_{N-k}}=0 }}
p_{B \vert XY\Lambda}\left(b_{1}, \ldots, b_{N}\vert 0, 0, \lambda\right). \label{eq:DefOmegaGeneral}
\end{align}
\end{subequations}

Considering that $\lambda \in \Lambda_{1}$, we can find upper bounds for $\vartheta_{\lambda}$ and $\omega_{\lambda}$ in a way similar to the one used in Supplementary Note 2. Given $r \in \{\alpha^{\lambda}_1, \ldots, \alpha^{\lambda}_k\}$, based on Eq.~\eqref{eq:AlphaApendGeneral} and by the fact that $l\kappa(\delta_{01}^{\lambda} + \delta_{10}^{\lambda} + \delta_{11}^{\lambda}) \le 1$,
\begin{equation}
\sum_{a_{1}, \ldots, \hat{a}_r,  \ldots, a_N} p_{A \vert XY\Lambda}\left(a_{1}, \ldots, a_{r}, 0, a_{r}, \ldots, a_{N} \vert 0, 1, \lambda\right) \leq \sqrt{\delta_{01}^{\lambda}} \leq \sqrt{\frac{1}{l\kappa}}.
\end{equation}
We then define $\Delta$ = $\{(a_{1},\ldots, a_{N}) \in \{0,1\}^{N} |  a_{\alpha^{\lambda}_1}=0 \lor \ldots \lor a_{\alpha^{\lambda}_k}=0\}$ and $\zeta_{01}$ as
\begin{equation}
\zeta_{01} \coloneqq p_{A \vert XY\Lambda}\left(\Delta \vert 0, 1, \lambda\right) = \sum_{\substack{a_{1},\ldots,a_{N} \\ a_{\alpha^{\lambda}_1}=0 \lor \ldots \lor a_{\alpha^{\lambda}_k}=0}} p_{A \vert XY\Lambda}\left(a_{1},  \ldots, a_{N} \vert 0, 1, \lambda\right) \leq k\sqrt{\delta_{01}^{\lambda}} \leq N\sqrt{\frac{1}{l\kappa}}.
\end{equation}
By the $(\varepsilon_A, \varepsilon_B)$-PD condition [Eq.~\eqref{eq:PD_epsilonA}],
\begin{equation}
 \frac{1}{2}\sum_{a_{1},\ldots, a_{N}}\left|p_{A \vert XY\Lambda}\left(a_{1}, \ldots, a_{N} \vert 0,0,\lambda\right)-p_{A \vert XY\Lambda}\left(a_{1}, \ldots, a_{N} \vert 0,1,\lambda\right)\right| \leq \varepsilon_{A}.
\end{equation}
Therefore, applying Proposition \ref{prop:ProbsGeneral}, we obtain the following upper bound for $\vartheta_{\lambda}$:
\begin{equation}
\vartheta_{\lambda} = \sum_{\substack{a_{1},\ldots,a_{N} \\ a_{\alpha^{\lambda}_1}=0 \lor \ldots \lor a_{\alpha^{\lambda}_k}=0}}p_{A \vert XY\Lambda}\left(a_{1}, \ldots, a_{N} \vert 0,0,\lambda\right) \leq \varepsilon_{A} + \zeta_{01} \leq \varepsilon_{A}+ N\sqrt{\frac{1}{l\kappa}}.
\end{equation}

The strategy for finding an upper bound for $\omega_{\lambda}$ is analogous. In fact, since we are considering $\lambda \in \Lambda_{1}$, Lemma \ref{lem:AlphaBetaGen} can be applied to ensure that $\left\{\bar{\alpha}^{\lambda}_{1}, \ldots, \bar{\alpha}^{\lambda}_{N-k}\right\} \subseteq\left\{\bar{\beta}^{\lambda}_{1}, \ldots, \bar{\beta}^{\lambda}_{N-k}\right\}$. Thus, for every $r \in\left\{\bar{\alpha}^{\lambda}_{1}, \ldots, \bar{\alpha}^{\lambda}_{N-k}\right\} \subseteq\left\{\bar{\beta}^{\lambda}_{1}, \ldots, \bar{\beta}^{\lambda}_{N-k}\right\}$, by Eq.~\eqref{eq:BetaBarApendGeneral},
\begin{equation}
\sum_{b_{1}, \ldots, \hat{b}_r,  \ldots, b_N} p_{B \vert XY\Lambda}\left(b_{1}, \ldots, b_{r-1}, 0, b_{r+1}, \ldots, b_{N} \vert 1, 0, \lambda\right) \leq \sqrt{\delta_{10}^{\lambda}} \leq \sqrt{\frac{1}{l\kappa}}.
\end{equation}
Let $\Delta$ = $\{(b_{1},\ldots, b_{N}) \in \{0,1\}^{N} |  b_{\bar{\alpha}^{\lambda}_1}=0 \lor \ldots \lor b_{\bar{\alpha}^{\lambda}_{N-k}}=0\}$ and $\zeta_{10}$ defined as
\begin{equation}\label{eq:Delta_sumBobGen}
\zeta_{10} \coloneqq p_{B \vert XY\Lambda}\left(\Delta \vert 1, 0, \lambda\right) = \sum_{\substack{b_{1},\ldots,b_{N} \\ b_{\bar{\alpha}^{\lambda}_1}=0 \lor \ldots \lor b_{\bar{\alpha}^{\lambda}_{N-k}}=0}} p_{B \vert XY\Lambda}\left(b_{1},  \ldots, b_{N}  \vert 1, 0, \lambda\right) \leq (N-k)\sqrt{\delta_{10}^{\lambda}} \leq N\sqrt{\frac{1}{l\kappa}}.
\end{equation}
By the $(\varepsilon_A, \varepsilon_B)$-PD condition [Eq.~\eqref{eq:PD_epsilonB}],
\begin{align}
&\frac{1}{2}\sum_{b_{1},\ldots,b_{N} } \left| p_{B \vert XY\Lambda}\left(b_{1}, \ldots, b_{N} \vert  0,0,\lambda\right) - p_{B \vert XY\Lambda}\left(b_{1}, \ldots, b_{N} \vert  1, 0, \lambda\right) \right| \leq \varepsilon_B.
\end{align}
Applying Proposition \ref{prop:ProbsGeneral}, we find the following upper bound for $\omega_{\lambda}$:
\begin{align}
&\omega_{\lambda} = \sum_{\substack{b_{1},\ldots,b_{N} \\ b_{\bar{\alpha}^{\lambda}_1}=0 \lor \ldots \lor b_{\bar{\alpha}^{\lambda}_{N-k}}=0}} p_{B \vert XY\Lambda}\left(b_{1}, \ldots, b_{N} \vert a_{1}, \ldots, a_{N}, 0,0,\lambda\right) \leq \varepsilon_B + \zeta_{10} \leq \varepsilon_{B}+ N\sqrt{\frac{1}{l\kappa}}.
\end{align}
To simplify, we use the notation:
\begin{subequations}
\label{ets}
\begin{align}
&\tilde{\varepsilon}_A = \varepsilon_{A}+ N\sqrt{\frac{1}{l\kappa}}, \\
& \tilde{\varepsilon}_B = \varepsilon_{B}+ N\sqrt{\frac{1}{l\kappa}},
\end{align}
\end{subequations}
and we will assume:
\begin{equation}\label{eq:LB_kappaN_SM}
\kappa > \frac{N^2}{l(1 - \varepsilon)^2},
\end{equation}
where $\varepsilon = \max\{\varepsilon_A, \varepsilon_B\}$. 
In this way,
\begin{equation}
\tilde{\varepsilon}_A = \varepsilon_{A}+ N\sqrt{\frac{1}{l\kappa}} < \varepsilon_{A} + (1 - \varepsilon) \le 1.
\end{equation}
An analogous reasoning applies to $\tilde{\varepsilon}_B$. Consequently, 
\begin{equation}
p_{H}^{N, \lambda} \le \vartheta_{\lambda} + \omega_{\lambda} - \omega_{\lambda}\vartheta_{\lambda}, 
\end{equation}
where $\vartheta_{\lambda} \in [0, \tilde{\varepsilon}_{A}]$ and $\omega_{\lambda} \in [0, \tilde{\varepsilon}_{B}]$. Therefore, it is easy to see that 
\begin{align}
p_{H}^{N, \lambda} & \leq \tilde{\varepsilon}_{A}+ \tilde{\varepsilon}_{B} - \tilde{\varepsilon}_{A}\tilde{\varepsilon}_{B}.
\end{align}
Moreover, the same bound is valid for $I_{\kappa}^{N}(p_{A B \vert XY})$,
\begin{align}
I_{\kappa}^{N}(p_{A B \vert XY}) &\le 4\sum_{\lambda \in \Lambda_{1}} p_{\Lambda}(\lambda)p_{XY|\Lambda}(0,0|\lambda) p_{H}^{N,\lambda} - 4l\kappa \sum_{\lambda \in \Lambda_{1}} p_{\Lambda}(\lambda)(\delta_{01}^{\lambda} + \delta_{10}^{\lambda} + \delta_{11}^{\lambda}) \notag \\
&\le 4(\tilde{\varepsilon}_{A}+ \tilde{\varepsilon}_{B} - \tilde{\varepsilon}_{A}\tilde{\varepsilon}_{B})\sum_{\lambda \in \Lambda_{1}} p_{\Lambda}(\lambda)p_{XY|\Lambda}(0,0|\lambda) -  4l\kappa \sum_{\lambda \in \Lambda_{1}} p_{\Lambda}(\lambda)(\delta_{01}^{\lambda} + \delta_{10}^{\lambda} + \delta_{11}^{\lambda}) \notag \\
& \le 4(\tilde{\varepsilon}_{A}+ \tilde{\varepsilon}_{B} - \tilde{\varepsilon}_{A}\tilde{\varepsilon}_{B})p_{XY}(0,0) \notag \\
& = \tilde{\varepsilon}_{A}+ \tilde{\varepsilon}_{B} - \tilde{\varepsilon}_{A}\tilde{\varepsilon}_{B}.
\end{align}



\section*{\texorpdfstring{Supplementary Note 5. Proof that HV models with PI, MI and $\delta$-OD do not reproduce all quantum correlations}{Supplementary Note 5. Proof that HV models with PI, MI and δ-OD do not reproduce all quantum correlations}}\label{app:OD}


We will consider quantum correlations in the bipartite Bell scenario in which each party chooses between $M+1$ measurements with two outcomes. We will label the inputs as $x, y \in \{0, \ldots, M\}$ and the outputs as $a, b \in \{0,1\}$.  

Under the assumptions of HV and MI, we first write the conditional distributions $p_{A,B|X,Y}(a,b|x,y)$ as 
\begin{eqnarray}
p_{A,B|X,Y}(a,b|x,y) &=& \sum_{\lambda} p_{\Lambda}(\lambda) p_{A,B|X,Y,\Lambda}(a,b|x,y,\lambda) \nonumber \\ &=& \sum_{\lambda} p_{\Lambda}(\lambda) p_{A|X,Y,\Lambda}(a|x,y,\lambda) p_{B|X,Y,A,\Lambda}(b|x,y,a,\lambda) \quad \forall a,b,x,y,
\end{eqnarray}
for hidden-variable state $\lambda$ and underlying probability densities $p_{A,B|X,Y,\Lambda}(a,b|x,y,\lambda) = p_{A|X,Y,\Lambda}(a|x,y,\lambda) p_{B|X,Y,A,\Lambda}(b|x,y,a,\lambda)$ and $p_{\Lambda}(\lambda)$. Now, the assumption of PI states that, for every $\lambda$, the marginal distribution of Alice's output $a$ is independent of Bob's input $y$, \textit{i.e.}, 
\begin{equation}\label{eq:PI-Alice}
p_{A|X,Y,\Lambda}(a|x,y,\lambda) = p_{A|X,\Lambda}(a|x,\lambda) \quad \forall \lambda, a, x.
\end{equation} 
Similarly, for every $\lambda$, the marginal distribution of Bob's output $b$ is independent of Alice's input $x$, \textit{i.e.},
\begin{eqnarray}\label{eq:PI-Bob}
\sum_a p_{A,B|X,Y,\Lambda}(a,b|x,y,\lambda) &=& \sum_a p_{A|X,\Lambda}(a|x,\lambda) p_{B|X,Y,A,\Lambda}(b|x,y,a,\lambda) \nonumber \\ 
&=& \sum_a p_{A|X,\Lambda}(a|x',\lambda) p_{B|X,Y,A,\Lambda}(b|x',y,a,\lambda) \quad \forall x, x', b, y, \lambda. 
\end{eqnarray}
Finally, the delta relaxation of OI, called $\delta$-OD, states that
\begin{equation}\label{eq:epsilon-OD_SM}
\frac{1}{2} \sum_b \big\vert p(b|x,y,a,\lambda) - p(b|x,y,a',\lambda) \big\vert \leq \delta \quad \forall x,y,a \neq a'.
\end{equation}

We consider the conditions of Hardy's ``ladder'' proof \cite{BBDH1997PRL} in this Bell scenario. Namely, we impose the conditions
\begin{subequations}\label{eq:Hardy-zeros}
\begin{align}
p_{A,B|X,Y}(0,0|0,0) &= 0, \\
p_{A,B|X,Y}(0,1|k,k-1) &= 0 \quad \forall k \in \{1,\ldots, M\},  \\
p_{A,B|X,Y}(1,0|k-1,k) &= 0 \quad \forall k \in \{1,\ldots, M\},  
\end{align}
\end{subequations}
and ask for the maximum value of $p_H = p_{A,B|X,Y}(0,0|M,M)$ under constraints \eqref{eq:Hardy-zeros}. For local hidden-variable theories, constraints \eqref{eq:Hardy-zeros} imply $p_H = 0$ \cite{BBDH1997PRL}. However, there are quantum correlations that achieve the value $p_H \rightarrow \frac{1}{2}$. Specifically, in the optimal quantum strategy, Alice and Bob share the state 
\begin{equation}
    |\psi \rangle = \alpha |0,0 \rangle - \beta |1, 1 \rangle,
\end{equation}
where $\alpha, \beta \in \mathbb{C}$ are parameters that depend on $M$ and obey $|\alpha|^2 + |\beta|^2 = 1$. Alice and Bob perform, on their half of the shared state, the measurements given by
\begin{subequations}
\begin{align}
A_k &= \sum_{j=0,1} | \pi_k^j \rangle \langle \pi_k^j|,  \\
B_k &= \sum_{j=0,1} | \sigma_k^j \rangle \langle \sigma_k^j|, 
 \end{align}
\end{subequations}
respectively, for $k \in \{0,\ldots, M\}$, with 
\begin{subequations}
\begin{align}
|\pi_k^0 \rangle &= \cos a_k|0 \rangle + \sin a_k | 1 \rangle, \quad |\pi_k^1 \rangle = - \sin a_k | 0\rangle + \cos a_k | 1 \rangle,  \\
|\sigma_k^0 \rangle &= \cos{b_k} | 0 \rangle + \sin{b_k} | 1 \rangle, \quad |\sigma_k^1 \rangle = - \sin{b_k}|0\rangle + \cos{b_k}|1\rangle,
\end{align}
\end{subequations}
for specific angles $a_k, b_k \in [0,2 \pi]$. To achieve constraints \eqref{eq:Hardy-zeros}, 
 \begin{subequations}
\begin{align}
\tan({a_k}) \cot({b_{k-1}}) &= - \frac{\alpha}{\beta} \quad \forall k \in \{1,\ldots, M\}, \\
\cot({a_{k-1}}) \tan({b_k}) &= - \frac{\alpha}{\beta} \quad \forall k \in \{1, \ldots, M \}, \\
\tan({a_0}) \tan({b_0}) &= \frac{\alpha}{\beta}.
 \end{align}
\end{subequations}
These conditions, together, give
\begin{equation}
\tan({b_M}) \tan({a_M}) = \left(\frac{\alpha}{\beta}\right)^{2M+1}.
\end{equation}
Optimising to achieve the maximum value for $p_H$, we obtain 
\begin{equation}
a_M = \arctan{\left[\left(\frac{\alpha}{\beta}\right)^{M+1/2}\right]}.
\end{equation}
This also gives that $\tan{a_k} = \tan{b_k}$ for all $k = 0,\ldots, M$, that is, $b_k = a_k + m \pi$ for each $k$. Finally, we obtain that 
\begin{equation}\label{eq:quantum-LadderHardy}
p_H = p_{A,B|X,Y}(0,0|M,M) = \max_{\alpha, \beta} \left(\frac{\alpha \beta^{2M+1} - \beta \alpha^{2M+1}}{\beta^{2M+1} + \alpha^{2M+1}} \right)^2 \rightarrow \frac{1}{2},
\end{equation}
where $p_H$ asymptotically grows as $1/2 - O(1/M)$ with the total number of settings $M+1$.

We now proceed to calculate the maximum value of $p_H = p_{A,B|X,Y}(0,0|M,M)$ under PI with $\delta$-OD, \textit{i.e.}, under Eqs.~(\ref{eq:PI-Alice}), (\ref{eq:PI-Bob}), and (\ref{eq:epsilon-OD_SM}). We first recognise that the problem of maximising $p_H = p_{A,B|X,Y}(0,0|M,M)$ is optimising a linear objective function under the linear constraints of PI for every $\lambda$ (that $\sum_b p_{A,B|X,Y,\Lambda}(a,b|x,y,\lambda)$ is independent of $y$ and similarly that $\sum_a p_{A,B|X,Y,\Lambda}(a,b|x,y,\lambda)$ is independent of $x$). Furthermore, the $\delta$-relaxation of OI can also effectively be formulated as a linear constraint (by considering both $\pm$ values for the absolute value constraint). As such, the optimal value is achieved at an extremal point of the corresponding correlation set, \textit{i.e.}, at some extremal value $\lambda^*$. For this $\lambda^*$, we first write constraints~(\ref{eq:Hardy-zeros}) as
\begin{subequations}
\begin{align}
p_{A|X,\Lambda}(0|0, \lambda^*) \cdot p_{B|X,Y,A, \Lambda}(0|0,0,0,\lambda^*) &= 0,\\
p_{A|X,\Lambda}(0|k,\lambda^*) \cdot p_{B|X,Y,A, \Lambda}(1|k,k-1,0,\lambda^*)&= 0 \quad \forall k \in \{1,\ldots, M\}, \\
p_{A|X, \Lambda}(1|k-1,\lambda^*) \cdot p_{B|X, Y, A, \Lambda}(0|k-1,k,1, \lambda^*) &= 0 \quad \forall k \in \{1,\ldots, M\}.
\end{align}
  \end{subequations}
Now, for each of the above equations, at least one of the terms in the product must be zero. We first establish that none of Alice's marginal terms $p_{A|X,\lambda^*}$ can be $0$ if one wants to achieve non-zero $p_H$. 

Let us assume that $M$ is odd without loss of generality, the argument for the case $M$ is even will follow analogously. Now suppose that $p_{A|X,\Lambda}(0|0, \lambda^*) = 0$, then it follows from the zero constraints that $p_{A,B|X,Y,\Lambda}(1,1|0,1,\lambda^*) = 1$, and similarly $p_{A,B|X,Y,\Lambda}(1,1|1,2,\lambda^*) = 1$, $p_{A,B|X,Y,\Lambda}(1,1|2,1,\lambda^*) = 1$, $p_{A,B|X,Y,\Lambda}(1,1|2,3,\lambda^*) = 1$, and so on until $p_{A,B|X,Y,\Lambda}(1,1|M, M-1, \lambda^*) = 1$ giving $p_H = p_{A,B|X,Y,\Lambda}(0,0|M,M, \lambda^*) = 0$ which is clearly not optimal. Similarly, suppose that $p_{A|X,\Lambda}(0|k, \lambda^*) = 0$ for some $k \in \{1,\ldots, M\}$, again it follows from the zero constraints that $p_{A,B|X,Y,\Lambda}(1,1|k,k+1, \lambda^*) = 1$, and from the same reasoning that either $p_{A,B|X,Y,\Lambda}(1,1|M-1, M, \lambda^*) = 1$ or $p_{A,B|X,Y,\Lambda}(1,1|M, M-1, \lambda^*) = 1$, and in either case we obtain $p_H = p_{A,B|X,Y,\Lambda}(0,0|M,M, \lambda^*) = 0$ which is not optimal. Strictly analogous reasoning following the chain of zeros gives that if $p_{A|X,\Lambda}(1|k-1, \lambda^*) = 0$ for some $k \in \{1,\ldots, M\}$ then again we end up with $p_H = p_{A,B|X,Y,\Lambda}(0,0|M,M, \lambda^*) = 0$. Therefore, we obtain that in order to satisfy constraints (\ref{eq:Hardy-zeros}) in this model, we must have for the optimal $\lambda^*$ that
\begin{subequations}
\begin{align}
p_{B|X,Y,A, \Lambda}(0|0,0,0,\lambda^*) &= 0, \\
p_{B|X,Y,A, \Lambda}(1|k,k-1,0,\lambda^*)&= 0 \quad \forall k \in \{1,\ldots, M\}, \\
p_{B|X, Y, A, \Lambda}(0|k-1,k,1, \lambda^*) &= 0 \quad \forall k \in \{1,\ldots, M\}.
\end{align}
  \end{subequations}
or, equivalently, by the normalisation condition that $p_{B|X,Y,A, \Lambda}(1|0,0,0,\lambda^*) = 1$, $p_{B|X,Y,A, \Lambda}(0|k,k-1,0,\lambda^*)= 1$ for all $k = 1,\ldots,M$ and $p_{B|X, Y, A, \Lambda}(1|k-1,k,1, \lambda^*) = 1$ for all $k = 1, \ldots, M$.

The delta relaxation of OI, $\delta$-OD, now gives that 
   \begin{subequations}
\begin{align}
p_{B|X,Y,A, \Lambda}(0|0,0,1,\lambda^*) &\leq \delta, \\
p_{B|X,Y,A, \Lambda}(1|k,k-1,1,\lambda^*)&\leq \delta \quad \forall k \in \{1,\ldots, M\},  \\
p_{B|X, Y, A, \Lambda}(0|k-1,k,0, \lambda^*) &\leq \delta \quad \forall k \in \{1,\ldots, M\}.
\end{align}
  \end{subequations}

Finally, for odd $M$, applying the PI condition for Bob (\ref{eq:PI-Bob}), we obtain that
   %
   \begin{eqnarray}\label{eq:SequenceP_hBound1}
        p_H = p_{A,B|X,Y,\Lambda}(0,0|M,M,\lambda^*) \leq p_{A|X,\Lambda}(0|M-1,\lambda^*) p_{B|X,Y,A,\Lambda}(0|M-1, M, 0,\lambda^*) &\leq& \delta p_{A|X,\Lambda}(0|M-1,\lambda^*), \nonumber \\
        p_{A|X,\Lambda}(0|M-1, \lambda^*) \leq p_{A|X,\Lambda}(0|M-3, \lambda^*) p_{B|X,Y,\Lambda}(0|M-3, M- 2, 0, \lambda^*)&\leq& \delta p_{A|X,\Lambda}(0|M-3, \lambda^*), \nonumber \\
        \vdots \nonumber \\
        p_{A|X,\Lambda}(0|2,\lambda^*) &\leq& \delta p_{A|X,\Lambda}(0|0, \lambda^*).
   \end{eqnarray}
   Analogously,
\begin{eqnarray}\label{eq:SequenceP_hBound2}
         p_{A,B|X,Y,\Lambda}(0,0|M,M-1,\lambda^*) = p_{A|X,\Lambda}(0|M, \lambda^*) \leq p_{A,B|X,Y,\Lambda}(0,0|M-2,M,\lambda^*) &\leq& \delta p_{A|X,\Lambda}(0|M-2,\lambda^*), \nonumber \\
        p_{A|X,\Lambda}(0|M-2, \lambda^*) \leq p_{A|X,\Lambda}(0|M-4, \lambda^*) p_{B|X,Y,\Lambda}(0|M-4, M- 3, 0, \lambda^*)&\leq& \delta p_{A|X,\Lambda}(0|M-4, \lambda^*), \nonumber \\
        \vdots \nonumber \\
        p_{A|X,\Lambda}(0|3,\lambda^*) &\leq& \delta p_{A|X,\Lambda}(0|1, \lambda^*), \nonumber\\
        p_{A|X,\Lambda}(0|1, \lambda^*) &\leq& \delta p_{A|X,\Lambda}(1|0, \lambda^*).
   \end{eqnarray}
   
Defining $q = (M+1)/2$, we obtain from Eqs.~\eqref{eq:SequenceP_hBound1} and \eqref{eq:SequenceP_hBound2} that 
\begin{eqnarray}\label{eq:pH-twoeqns}
p_H &\leq& \delta p_{A|X,\Lambda}(0| M-1, \lambda^*) \leq \delta^2 p_{A|X,\Lambda}(0|M-3,\lambda^*) \leq \ldots \leq \delta^q p_{A|X,\Lambda}(0|0,\lambda^*), \nonumber \\
p_H &\leq& \delta p_{A|X,\Lambda}(0|M-2, \lambda^*) \leq \delta^2 p_{A|X,\Lambda}(0|M-4,\lambda^*) \leq \ldots \leq \delta^{q} p_{A|X,\Lambda}(0|1,\lambda^*) \leq \delta^{q+1} p_{A|X,\Lambda}(1|0,\lambda^*) .
\end{eqnarray}

Adding the two inequalities in (\ref{eq:pH-twoeqns}) gives
\begin{eqnarray}
2 p_H &\leq& \delta^{q+1} p_{A|X,\Lambda}(1|0, \lambda^*) + \delta^q p_{A|X,\Lambda}(0|0,\lambda^*) \nonumber \\
&\leq & (\delta^q - \delta^{q+1}) p_{A|X,\Lambda}(0|0,\lambda^*) + \delta^{q+1} \leq \delta^q.
\end{eqnarray}
Finally, we obtain that, in models with PI and $\delta$-OD, the maximum value of $p_H = p_{A,B|X,Y}(0,0|M+1,M+1)$ is upper bounded as
\begin{equation}
p_H \leq \frac{\delta^q}{2},
\end{equation}
for $M = 2q - 1$ with integer $q \geq 1$. 

Comparing the value $\frac{\delta^q}{2}$ with the value $p_H = 1/2 - O(1/M)$ obtainable in quantum theory in Eq.~(\ref{eq:quantum-LadderHardy}), we see that, for any value of $\delta \in [0,1)$, there exists a quantum point that is outside the set of correlations obtainable in models with PI and $\delta$-OD. This completes the proof. 


\section*{Supplementary Note 6. Proof that the set of correlations produced by HV with MI, PI and complete OD is the set of nonsignaling correlations}


A correlation $p(a,b|x,y)$ has an HV model satisfying MI and PI, if there exists $p(\lambda)$ and $p(a,b|x,y,\lambda)$ such that
\begin{equation}
p(a,b|x,y) = \sum_{\lambda} p(\lambda)p(a,b|x,y,\lambda).
\end{equation}
where, by PI,
\begin{subequations}\label{eq:NonsignAppendA}
\begin{align}
& \sum_{b} p(a,b|x,y,\lambda) = p(a|x,\lambda) \quad \forall a,x,y,\lambda,
\\
& \sum_{a} p(a,b|x,y,\lambda) = p(b|y,\lambda) \quad \forall a,x,y,\lambda.
\end{align}
\end{subequations}
In this way, $p(a,b|x,y)$ is a convex sum of nonsignaling correlations, and therefore, in turn, is a nonsignaling correlation. On the other hand, given any nonsignaling correlation $p(a,b|x,y)$, defining $ p(a,b|x,y,\lambda) = p(a,b|x,y)$, we have that $ p(a,b|x,y,\lambda)$ satisfies Eqs.~\eqref{eq:NonsignAppendA}. Therefore, $p(a,b|x,y)$ has a hidden variable model satisfying MI and PI.
%


















%


\end{document}